\documentclass{article}

\usepackage{PRIMEarxiv}

\usepackage[utf8]{inputenc} 
\usepackage[T1]{fontenc}    
\usepackage{url}            
\usepackage{booktabs}       
\usepackage{amsfonts}       
\usepackage{nicefrac}      
\usepackage{microtype}      
\usepackage{lipsum}
\usepackage{fancyhdr}       
\usepackage{graphicx}       
\graphicspath{{media/}}     
\usepackage[round]{natbib}
\usepackage{makecell}
\usepackage{url}
\usepackage[hidelinks]{hyperref}
\usepackage[small]{caption}
\usepackage{amsmath}
\usepackage{amsthm}
\usepackage{algorithm}
\usepackage{algorithmic}
\usepackage{amssymb}
\usepackage{tikz}
\usepackage{pifont}
\usepackage{complexity}
\usepackage{soul}
\usepackage{xcolor}

\usepackage[textwidth=30mm]{todonotes}

\newtheorem{theorem}{Theorem}

\newtheorem{corollary}[theorem]{Corollary}
\newtheorem{proposition}[theorem]{Proposition}

\newtheorem{definition}{Definition}
\newtheorem{example}{Example}

\newcommand{\I}{\mathcal{I}} 
\renewcommand{\C}{\mathcal{C}} 
 
\newcommand{\cEJR}{\text{cohEJR}_{\C}}

\newcommand{\cJR}{\text{cohJR}_{\C}}

\newcommand{\agrEJR}{\text{agrEJR}_{\C}}
\newcommand{\cMES}{\text{MES}_{\C}}
\renewcommand{\leq}{\leqslant}
\renewcommand{\geq}{\geqslant}
\newcommand{\prof}[1]{\boldsymbol{#1}}
\renewcommand{\vec}[1]{\boldsymbol{#1}}
\newcommand{\Agr}[2]{\text{\rm Agr}(#1,#2)}
\newcommand{\fixDec}[3]{{#1}[{#2\leftarrow #3}]}
\newcommand{\sat}[2]{u_{#1}(#2)}

\newcommand{\myparagraph}[1]{\vspace*{3pt}\noindent\textbf{#1}}

\newcommand{\cmark}{\mbox{\ding{51}}}%
\renewcommand{\vec}[1]{\boldsymbol{#1}}

\theoremstyle{plain}

\newtheorem{myex}{Example} 
\renewenvironment{example}{\begin{myex}\rm}{\hfill$\vartriangle$\end{myex}}

\title{Proportionality for Constrained Public Decisions
}

\author{
  Julian Chingoma\\
  ILLC \\
  University of Amsterdam \\
  Amsterdam, The Netherlands\\
  j.z.chingoma@uva.nl\\
  0000-0002-6491-3779\\
   \And
  Umberto Grandi \\
  IRIT \\
  Univeristé Toulouse Capitole \\
  Toulouse, France\\
  umberto.grandi@irit.fr\\
  0000-0002-1908-5142\\
  \And
  Arianna Novaro\\
  CES \\
  Université Paris 1 Panthéon-Sorbonne \\
  Paris, France\\
  arianna.novaro@univ-paris1.fr\\
  0000-0003-3443-1530\\
}

\begin{document}
\maketitle

\begin{abstract}
We study situations where a group of voters need to take a collective decision over a number of public issues, with the goal of getting a result that reflects the voters’ opinions in a proportional manner. 
Our focus is on interconnected public decisions, where the outcome on one or more issues has repercussions on the acceptance or rejection of other issues in the agenda. 
We show that the adaptation of classical justified-representation axioms to this enriched setting are always satisfiable only for restricted classes of public agendas. 
We adapt well-known proportional decision rules to take the structure of the public agenda into account, and we show that they match justified-representation properties in approximation on a class of expressive constraints.
We also identify another path to achieving proportionality on interconnected issues via an adaptation of the notion of priceability.
\end{abstract}

\section{Introduction}\label{sec:intro}

Consider a municipality that is planning the renovation of one of its main squares in the historical centre. A project of this calibre implies taking numerous decisions, such as on whether to plant some trees (and of which kind), whether to add a fountain, some benches and tables, a bike-sharing station, a statue (and from which artist), and so on. Some of the potential configurations of the square may also be unfeasible, either because they would go over the allocated budget for the project, or because there may not be enough physical space to accommodate all of the desired features. 

Given the impact that such a renovation would have on the everyday lives of its citizens, the municipality may decide to set up a \emph{participatory design} process,
in order to make sure that different preferences and needs are taken into account in the final design.\footnote{Our example evokes the participatory process that was set up by the Municipality of Madrid in order to collectively redesign \emph{Plaza de Espa\~{n}a}---although the project was later criticized by the opposition due to its low participation rate \citep{elpais2017}.} This is just one of many real-world examples of public decisions where \emph{(i)} the final outcome should strive to be \emph{representative} of the views of the different stakeholders involved, while also \emph{(ii)} having some external \emph{constraints} that rule out some possible outcomes as unfeasible. 

Regarding our point \emph{(ii)} above, indeed, the square re-design problem can be seen as a (generalized) \emph{participatory budgeting} instance (cf. the recent survey by \citeauthor{reyM2023PBsurvey}, \citeyear{reyM2023PBsurvey}), where the implementation of one project may be conditional on the acceptance, or rejection, of another, all the while satisfying some monetary and spatial constraints. We can also think of the collective \emph{product configuration} problem  \citep{AstesanaConfiguration2010}, where a group of people may need to collectively choose the features of a given product (e.g, a group of friends booking an accommodation for their holidays), knowing that some value combinations may be unavailable or out of stock. Finally, in a \emph{committee election} \citep{LacknerS23mwv} we may need to fill the positions in a committee with some candidates by also respecting various diversity constraints. 

The latter example is an instance of an important sub-setting of public decisions, namely one where the issues at stake all have a \emph{binary} domain: namely, the decision-makers are asked to express their preferences as an acceptance or a rejection of each issues (or candidate). These kind of scenarios may occur also in the context of \emph{multiple referenda}, where the public votes directly on the resolution of political issues, and for \emph{group activity planning}, where a group of individuals has to choose, as a collective, the activities that the entire group shall partake in \citep{langX16HBCOMSOC2016}. 

If we now turn to our point \emph{(i)} above, among the numerous interpretations of \emph{fairness} that have been put forward to make sure that the outcome appropriately reflect the diverse views of the participants, we will focus on the notion of \emph{proportional representation}. Proportionality features prominently in many collective choice settings such as that of \emph{apportionment} \citep{Balinski2005-xu} and the aforementioned committee elections \citep{LacknerS23mwv}. Indeed, even when zooming in on the problem of \emph{public decisions} \citep{skowronGoreckiAAAI2022public}, the goal of producing collective outcomes that proportionally reflect the opinions of the voter population has been drawing increasing attention in recent years \citep{freemanEtAlIJCAI2020variable,MasarikP024generalProp}. 
However, a component that has so far not received much attention in this growing literature on proportionality
is the presence of \emph{constraints} that restrict the possible outcomes that can be returned, i.e., the combination of our points \emph{(i)} and \emph{(ii)} above. In this paper, we thus focus on answering the question of what one may do when outcomes that would satisfy classical proportionality axioms---and thus be considered fair---are no longer feasible due to the presence of constraints.


Our contributions chart the extent to which proportionality can be ensured in a constrained public-decision setting. First, we use the notion of feasible group deviations as a building block that allows the adaptation of existing proportionality axioms---that are based on varying public-decision interpretations of justified representation---to decisions with constraints. For each of our axioms, we show that although it is challenging to satisfy these properties in general constrained instances, when one hones in on a restricted---yet highly expressive---class of constraints, we can achieve proportionality guarantees that represent approximations of desirable justified-representation axioms. In doing so, we also define novel adaptations of recently studied decision rules to our public-decision setting with constraints, namely the method of equal shares (\emph{MES}) and the \emph{MeCorA} rule. Finally, we adapt the priceability notion from the literature on committee elections, which provides another promising route to introduce proportionality into public decisions under constraints.
A detailed summary of our results can be found in Table~\ref{tab:axiomTable}.

This paper is structured as follows. We begin by discussing related work in Section~\ref{sec:related}. We then continue with Section~\ref{sec:model} where we detail the constrained public-decision model, discuss the existing axioms of justified representation for public decisions, and also present the notion of deviating groups. Then each of Sections \ref{sec:justified_cohesive} and \ref{sec:justified_agreement} deals with a particular public-decision interpretation of justified representation, namely via cohesiveness and via agreement. Before concluding in Section~\ref{sec:conclusion}, we deal with our constrained version of the priceability axiom in Section~\ref{sec:priceable}.

\section{Related Work}\label{sec:related}

We begin by noting that our constrained public-decision model closely resembles that of \emph{judgment aggregation} and it also naturally fits into the area of collective decisions in \emph{combinatorial domains} (see \citeauthor{EndrissHBCOMSOC2016}, \citeyear{EndrissHBCOMSOC2016}, and \citeauthor{langX16HBCOMSOC2016}, \citeyear{langX16HBCOMSOC2016}, for general introductions to these two topics, respectively). 

Most relevant to our paper is the recent work conducted on fairness in the context of public decisions without constraints \citep{conitzerFS2017fairPD,freemanEtAlIJCAI2020variable,skowronGoreckiAAAI2022public}. \citet{conitzerFS2017fairPD} focused on individually proportional outcomes, thus, our work more closely aligns with that of \citet{freemanEtAlIJCAI2020variable} and \citet{skowronGoreckiAAAI2022public} who import to public decisions the notion of \emph{justified representation} \citep{azizEtAlSCW2017justified,fernandezEtAlAAAI2017PJR,petersSkowronEC2020welfarism} from the literature on \emph{committee elections} \citep{LacknerS23mwv}. 
Proportionality has also been studied in models of \emph{sequential decision-making} that are relevant to our own as they can be seen as generalisations of the public-decision model without constraints \citep{bulteauHPRT2021perpetualJR,chandakARXIV2024sequential,lackner2020perpetual}. Amongst these sequential decision-making papers, those of \citet{bulteauHPRT2021perpetualJR} and \citet{chandakARXIV2024sequential} relate to our work the most as they also implement justified-representation notions. In particular, \citeauthor{chandakARXIV2024sequential} study a model similar to ours, where decisions involve multiple alternatives. The key distinction is that our work incorporates feasibility constraints on the final outcomes.

Proportionality under constrains have been studied extensively by \citet{MasarikP024generalProp} in a general social-choice model that can model both the unconstrained and constrained versions of public-decisions, along with a variety of other social choice domains. Their model elicits approval ballots from voters over a set of alternatives, and imposes a feasibility constraint on potential outcomes. 
The most notable differences with the model we work with are that $(1)$ while our proportionality axioms—based on the idea of \emph{feasible voter deviations}—are structurally similar to those of \citeauthor{MasarikP024generalProp}, our axioms demand less from voter groups when evaluating their deserved representation, making them arguably more natural but also harder to satisfy; and $(2)$ by introducing properties tailored to our specific setting, we define and analyse constrained public-decision rules not considered by \citeauthor{MasarikP024generalProp}, such as the method of equal shares (an open question explicitly mentioned by \citeauthor{MasarikP024generalProp}).

In a similar line, \citet{mavrovMSEC2023fair} adapted justified representation for the committee elections model with arbitrary constraints, mostly focused on the notion of the core but also defining a version of extended justified representation. We refer to the discussion by \cite{MasarikP024generalProp} for illustrating the limitations of their approach.

In related fields, previous work studied proportionality in various models that feature collective choices on interconnected propositions: the \emph{belief merging} setting \citep{haretEtAlAAAI2020belief}, interdependent binary issues via \emph{conditional ballots} \citep{brillEtAlIJCAI2023prop}, committee elections with synergies \citep{izsakEtAlAAAI2018synergies}, participatory budgeting with project interactions \citep{jainEtAlIJCAI2020interactions}, and approval-based shortlisting with constraints, in a model of judgment aggregation by \citet{chingomaEtAlAAMAS2022simulate}.

\begin{table*}[t]
\centering  
\begin{tabular}{cl}
\toprule
  Proportionality  &  Satisfiability results \\ 
\midrule
$\alpha\text{-}{\cEJR}\text{-}\beta$   & $\cmark$ $\cMES$ with $\lambda_{\text{unit}}$ for unconstrained elections with $\alpha = \nicefrac{2}{(\max_{t\in[m]}|D_{t}|)}$ and $\beta = 1$ (Prop.~\ref{prop:MES_satisfies_EJR-1}) \\
$\Uparrow$ &  \\

$\cEJR$-1    &  \makecell[l]{$\cmark$ $\cMES$ with $\lambda_{\text{unit}}$ for binary unconstrained elections (Prop.~\ref{prop:MES_satisfies_EJR-1}) \\ $\cmark$ for binary elections with $m = 2$ (Prop.~\ref{prop:c-EJR-1_always_satisfiable_m=2})} \\
$\Uparrow$ &  \\
$\cEJR$   & \makecell[l]{Not always satisfiable \citep{chandakARXIV2024sequential}\\$\cmark$ when $|\C| = 2$ and $\C$ has $\text{NFD}$ (Prop.~\ref{prop:c-EJR_satisfiable_|C|=2_A1})\\ $\cmark$ for binary elections with $m \leq 3$ and $\C$ has NFD (Prop.~\ref{prop:c-EJR_satisfiable_|I|<=3_A1})} \\
$\Downarrow$ &  \\ 
$\cJR$   & Not always satisfiable (Prop.~\ref{prop:c-JR_may_not_exist_A1}) \\
\midrule
$\alpha\text{-}{\agrEJR}\text{-}\beta$    & \makecell[l]{$\cmark$ $\cMES$ with $\lambda_{\text{fix}}$ for binary elections, $k$-restrictive $\C$ with $\alpha = \nicefrac{1}{k}$ and $\beta = 1$ (Thm.~\ref{thm:c-MES-fix_satisfies_1/k-agrEJR-1})\\ $\cmark$ Greedy $\text{MeCorA}_{\C}\text{-}(k-1)$ when $\C$-consistent, $k$-restrictive $\C$ 
(Thm.~\ref{thm:greedyMeCorA_approx_satisfies_agrEJR}) \\ $\cmark$ LS-$\text{PAV}_{\C}$ when $\C$-consistent, $k$-restrictive $\C$ with  $\alpha = \nicefrac{2}{(k+1)}$ (Thm.~\ref{thm:lsPAV_approx_satisfies_agrEJR})} \\
$\Uparrow$ &  \\
$\agrEJR$    & Not always satisfiable even when $\C$ has NFD (Cor.~\ref{prop:agrEJR_not_always_satisfiable_NFD}) \\
\midrule
Priceability    & $\cmark$ $\text{MeCorA}_{\C}$ (Prop.~\ref{prop:MeCorA-c_is_priceable})  \\
\bottomrule
\end{tabular}

\caption{Table summarising our main results. The first column presents the constraints-adapted notions of proportionality we consider, and their relative implications. The second column presents results on their satisfiability: $\cmark$ symbol identifies results where the axiom is satisfied by one of our proposed rules, and mentions the relevant theorem or reference.}
\label{tab:axiomTable}
\end{table*}

\section{The Model}\label{sec:model}

A finite set of $n$ voters $N = \{1,\ldots,n\}$ has to take a collective decision on a finite set of $m$ issues $\I = \{a_{1},\ldots,a_{m}\}$. Typically, the public decision setting considers only two available decisions per issue, but we instead adopt the following more general setup. 
Taking $X$ as a set of \emph{alternatives}, each issue $a_{t}\in\I$ is associated with its own finite set of alternatives called a \emph{domain} $D_{t} = \{d_{t}^{1},d_{t}^{2},\ldots\}\subseteq X$ where $|D_{t}| \geq 2$ holds for all $t\in[m]$ with $X$. The design decision to go beyond binary issues is motivated by the wider real-life applicability of this model.
%
Each voter $i\in N$ submits a \emph{ballot} $\vec{b}_{i} = (\vec{b}_{i,1},\ldots,\vec{b}_{i,m}) \in D_{1}\times\ldots\times D_{m}$ 
where an entry $\vec{b}_{i,t}$ indicates that voter $i$ chooses alternative $c$ for the issue~$a_t$. 
A \emph{profile} $\prof{B} = (\vec{b}_{1},\ldots,\vec{b}_{n}) \in (D_{1}\times\ldots\times D_{m})^{n}$ is a vector of the $n$ voters' ballots. An \emph{outcome} $\vec{w} = (w_{1},\ldots,w_{m}) \in D_{1}\times\ldots\times D_{m}$ is then a vector providing a decision for every issue at stake. 

We focus on situations where some \emph{constraints} limit the set of possible collective outcomes: we denote by $\C \subseteq D_{1}\times\ldots\times D_{m}$ the non-empty set of \emph{feasible} outcomes. We write $(\prof{B}, \C)$ to denote an \emph{election instance}. By a slight abuse of notation we also refer to $\C$ as the \emph{constraint}, and thus, we refer to elections instances where $\C = D_{1}\times\ldots\times D_{m}$ as \emph{unconstrained} election instances.\footnote{Note that while we work formally with the constraint being an enumeration of all feasible outcomes, in practice, it is often possible to represent the set of feasible outcomes in more concise forms---via the use of formulas of \emph{propositional logic}, for example---to help with parsing said constraint and/or speed up computation by exploiting the constraint's representation structure.}
Note that voter ballots need not be consistent with the constraints, i.e., for an election instance $(\prof{B}, \C)$, we do not require that $\vec{b}_{i}\in \C$ for all voters $i\in N$.\footnote{This assumption takes our model closer to the particular model of judgment aggregation where the constraints on the output may differ from the constraints imposed on the voters' input judgments \citep{EndrissAAMAS2018,chingomaEtAlAAMAS2022simulate}.}  
%
This assumption is common in a number of settings of social choice. In committee elections, voters can approve more candidates than the committee target size while in participatory budgeting, the sum of the costs of a voter's approved projects may exceed the instance's budget.
When deviating from this assumption we may explicitly require that voter ballots be constraint-consistent. 
For our setting, we argue that applying constraints only on the outcome of the collective decision helps capture real-world, constrained decision-making scenarios where either the constraint is \emph{uncertain} when voters submit their ballots, or possibly, the voting process becomes more burdensome for voters as they attempt to create ballots with respect to a (possibly difficult to understand) constraint. 

\begin{example}
    A group of friends is deciding on the travel destinations of their shared holiday across France, visiting different regions. On a booking platform, there are a certain number of cities that can be selected per region such as: Toulouse, Carcassonne and Montpellier in Occitania; Rouen and Caen in Normandy; Rennes and Brest in Brittany; and so on. Each friend has a preferred combination of cities and their collective itinerary is subject to factors such as their travel budget or the available connections between cities. However, as costs and connections may change significantly on a day-to-day basis, it may be unclear which combination of cities are affordable.  Therefore, it is not reasonable to impose the default requirement that voter ballots are constraint-consistent.
\end{example}

At times, we shall restrict to  election instances where $D_t = \{0,1\}$ for every issue $a_t$. We refer to such cases as \emph{binary} election instances. We explicitly state whether any result hinges on the restriction to binary instances. Given an outcome $\vec{w}$ for a binary instance, we define the vector $\vec{\bar w} = ({\bar w}_{1},\ldots,{\bar w}_{m}) $ as $\bar w_t = 1-w_t$ for all issues $a_t\in \I$.

Let us now introduce some useful notation.
Consider an outcome $\vec{w}$, a set of issues $S\subseteq \I$ and some vector $\vec{v} = (v_{1},\ldots,v_{m})\in D_{1}\times\ldots\times D_{m}$ (that can be interpreted as either an outcome or a voter's ballot). We write $\fixDec{\vec{w}}{S}{\vec{v}} = (w_{1}',\ldots,w_{m}')$ where $w_{t}' = w_{t}$ for all issues $a_t\in \I\setminus{S}$ and $w_{t}' = v_{t}$ for all issues $a_j\in S$. In other words, $\fixDec{\vec{w}}{S}{\vec{v}}$ is the resultant vector of updating outcome $\vec{w}$'s decisions on the issues in $S$ by fixing them to those of vector $\vec{v}$. 
For a given issue $a_t \in \I$ and a decision $d\in D_t$, we use $N(a_t,d) = \{i\in N \mid \vec{b}_{i,t} = d\}$ to denote the set of voters that agree with decision $d$ on issue $a_{t}$. 
Given two vectors $\vec{v}, \vec{v}' \in D_{1}\times\ldots\times D_{m}$, we denote the \emph{agreement} between them by $\Agr{\vec{v}}{\vec{v}'} = \{a_t\in \I\mid v_{t} = {v}_{t}'\}$. 

We define the \emph{satisfaction} that a voter $i$ obtains from an outcome $\vec{w}$ as $\sat{i}{\vec{w}} = |\Agr{\vec{b}_{i}}{\vec{w}}|$, i.e., the number of decisions on which the voter $i$ is in agreement with outcome $\vec{w}$. We adopted this notion of satisfaction as it is a common choice within the literature on public decisions \citep{freemanEtAlIJCAI2020variable,skowronGoreckiAAAI2022public}, which is in turn grounded on a similar assumption from work on committee elections \citep{LacknerS23mwv}.

\subsection{Proportionality via Justified Representation}\label{subsec:justified}

This section starts with the observation that classical notions of proportionality fall short when considering interconnected decisions (see Example~\ref{exm:no_fair_outcome_EJR}), and then follows with our adaptations of such axioms for constrained environments. 

Ideally, when looking to make a proportional collective choice, we would like to meet the following criteria: a group of similarly-minded voters that is an $\alpha$-fraction of the population should have their opinions reflected in an $\alpha$-fraction of the $m$ issues. 
%
In the setting of committee elections, this is formally captured via the justified representation axioms, with one of the most widely studied being \emph{extended justified representation (EJR)} \citep{azizEtAlSCW2017justified}. In the setting of public decisions there are two adaptations that have been studied. The first approach intuitively states that ``a group of voters that agree on an $\alpha$-fraction of the issues in $\I$ and are $\alpha$-fraction of the voter population, should have some control over an $\alpha\cdot|\I|$ number of the issues in $\I$'' \citep{chandakARXIV2024sequential,freemanEtAlIJCAI2020variable}. In this approach, the requirements on the voter groups 
are captured by the notion of \emph{cohesiveness} and so we refer to this version of EJR as \emph{cohesiveness-EJR}.
The second approach intuitively states that ``a group of voters that agree on a set of issues $T$ and represent an $\alpha$-fraction of the voter population, should have some control over an $\alpha\cdot |T|$ number of the total issues in $\I$'' \citep{chandakARXIV2024sequential,MasarikP024generalProp,skowronGoreckiAAAI2022public}. We refer to it as \emph{agreement-EJR}.
 Observe that cohesiveness-EJR is stronger than, and implies, agreement-EJR.

Meeting the ideal outlined by both of these notions is not easy in our setting, as the constraint $\C$ could rule out a seemingly fair outcome from the onset.

\begin{example}\label{exm:no_fair_outcome_EJR}
Suppose there are two binary issues $\I = \{a_1, a_2\}$ with constraint $\C  = \{(1,0),(0,1)\}$. Then suppose there are two voters $N = \{1,2\}$ with ballots $\vec{b}_{1} = (1,0)$, and $\vec{b}_{2} = (0,1)$. 
Here, both aforementioned EJR interpretations require each voter to obtain at least $1$ in satisfaction, i.e., deciding half of the two issues at hand. This would be possible by selecting outcomes $(1,1)$ or $(0,0)$, which are however not feasible according to the constraints. Thus, there exists no feasible outcome that provides agreement-EJR or cohesiveness-EJR as one voter $i\in \{1,2\}$ will have satisfaction $\sat{i}{\vec{w}} = 0$ for any outcome $\vec{w}\in \C$. 
\end{example}

\noindent
Example~\ref{exm:no_fair_outcome_EJR} shows that a voter group that is an $\nicefrac{1}{2}$-fraction of the population may lay claim to deciding an $\nicefrac{1}{2}$-fraction of the issues, but in doing so, they may be resolving, or influencing the decision on, a larger portion of the issues than they are entitled to. This is an issue we must take into account when defining proportionality when there are constraints.

\subsection{Constrained Deviations}\label{sec:deviations}

To define proportionality axioms that accommodate constraints, we seek meaningful ways to identify, given an outcome~$\vec{w}$, those voter groups that are happier with the selection of an outcome that is different from $\vec{w}$. 
We formalize this by the notion of 
deviation, in line with related work \citep{azizEtAlSCW2017justified,MasarikP024generalProp,HaretK0S24equilibria}.

\begin{definition}[$(S,\vec{w},\C)$-deviation]
Given election instance $(\prof{B},\C)$ and outcome $\vec{w}\in \C$, a set of voters $N'\subseteq N$ has an \emph{$(S,\vec{w},\C)$-deviation} if $\emptyset\neq S \subseteq \I$ is a set of issues such that all of the following hold:
    \begin{itemize}
        \item $S \subseteq\Agr{\vec{b}_{i}}{\vec{b}_{j}} \text{ for all } i,j\in N'$ (the voters agree with each other on the decisions on all issues in $S$).
        \item $S\subseteq\I\setminus{\Agr{\vec{b}_{i}}{\vec{w}}}$ for all $i\in N'$  (each voter disagrees with outcome $\vec{w}$'s decisions on all issues in $S$).
        \item $\fixDec{\vec{w}}{S}{\vec{b}_{i}} \in \C$ for all $i\in N'$ (fixing outcome $\vec{w}$'s decisions on issues in $S$, so as to agree with the voters in $N'$, induces a feasible outcome).
    \end{itemize}
\end{definition}

\noindent
We provide an example of a deviation as defined above with the following example.

\begin{example}
    Suppose there are three binary issues $\I = \{a_1, a_2, a_3\}$ with constraint $\C  = \{(0,1,1),(1,0,1)\}$, and three voters $N = \{1,2,3\}$ with ballots $\vec{b}_{1} = (1,0,0)$, $\vec{b}_{2} = (1,0,0)$ and $\vec{b}_{3} = (0,0,0)$. Take outcome $\vec{w}=(0,1,1)$. Voters $N' = \{1,2\}$ agree on issues in  $S = \{a_1, a_2\}$ and they disagree with $\vec{w}$ on the decisions for issues in $S$ and $\fixDec{\vec{w}}{S}{\vec{b}_{1}} = \fixDec{\vec{w}}{S}{\vec{b}_{2}} = (1,0,1) \in \C$. Hence, voters $1$ and $2$ have an $(S,\vec{w},\C)$-deviation to outcome $(1,0,1)$.
\end{example}

\noindent
Intuitively, given an outcome $\vec{w}$, a voter group having an $(S,\vec{w},\C)$-deviation indicates the presence of another feasible outcome $\vec{w}^{*} \neq \vec{w}$ where every group member would be better off. Thus, our goal in providing a fair outcome reduces to finding an outcome where every group of voters that has an $(S,\vec{w},\C)$-deviation is sufficiently represented.

\section{Justified Representation with Cohesiveness}\label{sec:justified_cohesive}

 To adapt cohesiveness-EJR to public decisions with constraints, we start from the notion of cohesiveness 
 introduced in proportionality studies in the committee elections literature.
 We say that a voter group $N'$ is \emph{$T$-agreeing} for some set of issues $T\subseteq \I$ if $T \subseteq\Agr{\vec{b}_{i}}{\vec{b}_{j}}$ holds for all voters $i,j\in N'$ and then we define cohesiveness as the following:

\begin{definition}[$T$-cohesiveness]
For a set of issues $T\subseteq \I$, we say that a set of voters $N'\subseteq N$ is \emph{$T$-cohesive} if $N'$ is $T$-agreeing and it holds that $|N'| \geq |T|\cdot\nicefrac{n}{m}$.
\end{definition}

\noindent 
Using $T$-cohesiveness, we can define our version of EJR for public decisions with constraints.

\begin{definition}[$\cEJR$]\label{def:cEJR}
Given an election $(\prof{B}, \C)$, an outcome $\vec{w}$ provides \emph{ $\cEJR$} if for every $T$-cohesive group of voters $N'\subseteq N$ for some $T\subseteq \I$ with an $(S,\vec{w},\C)$-deviation for some non-empty $S\subseteq T$, there exists a voter $i\in N'$ such that:
\[\sat{i}{\vec{w}} \geq |T|.\]
\end{definition}

\noindent
Intuitively, $\cEJR$ deems an outcome to be unfair if there exists a $T$-cohesive voter group with $(i)$ none of its group members having at least $|T|$ in satisfaction, and $(ii)$ changing outcome $\vec{w}$'s decisions on some of the issues in $T$, specifically those in a subset $S\subseteq T$, leads to some other feasible outcome.

\begin{example}\label{exm:cEJR}
    Suppose there are three issues $\I = \{a_1, a_2, a_3\}$ with domains $D_{t} = \{d_{1},d_{2},d_{3}\}$ for all $t\in\{1,2,3\}$. Then take the constraint to be $\C  = \{(d_{1},d_{2},d_{3}),(d_{2},d_{1},d_{3})\}$. Then, suppose there are three voters $N = \{1,2,3\}$ with ballots $\vec{b}_{1} = (d_{1},d_{2},d_{1})$, $\vec{b}_{2} = (d_{1},d_{2},d_{2})$ and $\vec{b}_{3} = (d_{3},d_{3},d_{2})$. Note that voters $1$ and $2$ form a $T$-cohesive group for the set of issues $T = \{a_1, a_2\}$ while voter group $\{2,3\}$ is a $T$-cohesive group for $T=
    \{a_3\}$. Take outcome $\vec{w}=(d_{2},d_{1},d_{3})$. Each of voters $N' = \{1,2\}$ has a satisfaction of $0$ and the group has an $(S,\vec{w},\C)$-deviation (with $S=T$) to the alternative outcome $\vec{w}^{*}=(d_{1},d_{2},d_{3})$ thus, outcome $\vec{w}$ does not provide $\cEJR$. For outcome $\vec{w}^{*}$, see that while the voters are a $T$-cohesive group for $T=\{a_3\}$, they do not have an $(S,\vec{w}^{*},\C)$-deviation to a feasible outcome. Thus, it follows that outcome $\vec{w}^{*}$ provides $\cEJR$.
\end{example}

\noindent
We now discuss our proposal in the context of recent adaptations of proportionality axioms to constraints.
We start with the work of \citet{MasarikP024generalProp},
precisely with the axioms of \emph{Base EJR} and \emph{EJR} (Definitions 2 and 4 in the cited paper). 
Consider the election instance in Example~\ref{exm:cEJR}. In this example, we argue that voters $\{1,2\}$ is $T$-cohesive for issues $T=\{a_1,a_2\}$ and is witness to a violation of $\cEJR$ by the outcome $\vec{w}=(d_{2},d_{1},d_{3})$ (at least one of them deserves $|T|=2$ in satisfaction). 
However, in the context of the \emph{Base EJR} axiom of \citet{MasarikP024generalProp}, voter group $\{1,2\}$ would not qualify as a group deserving of representation. To see this, we find a partial outcome within the feasible ones in $\C$, namely $(d_{2},\_,d_{3})$, such that both of the following hold: $(1)$ the group cannot `complete' the partial outcome with decisions that they agree on, and $(2)$ the group size is not strictly larger than $n\cdot \nicefrac{|T|}{(2+|T|)}$ (where $2$ is the number of decided issues in the partial outcome $(d_{2},\_,d_{3})$). 
The same holds for the stronger EJR axiom by \citet{MasarikP024generalProp}, which asks to consider only partial outcomes agreeing with $\vec{w}=(d_{2},d_{1},d_{3})$
($(d_{2},\_,d_{3})$ still qualifies). 
These considerations illustrate that the axioms of \citet{MasarikP024generalProp} are more demanding of voter groups when assessing whether they deserve representation. 
We argue that our proposal of $\cEJR$ axiom (and other notions to follow) are more intuitive and appropriately lenient when identifying groups worthy of representation but we do note that this makes our notion much more challenging to satisfy.
Similar reasoning applies to Restrained EJR \cite{mavrovMSEC2023fair}. We direct readers to \citep{MasarikP024generalProp} for an in-depth discussion (Section $3.2$ of their paper) on the similarities between Restrained EJR and Base EJR.

Next, we ask how computationally demanding it is to check whether some outcome provides $\cEJR$. But first, we look at the following computational question associated with $(S,\vec{w},\C)$-deviations: given an election instance $(\prof{B},\C)$ and an outcome $\vec{w}\in \C$, the problem is to find all groups of voters with an $(S,\vec{w},\C)$-deviation.

\begin{proposition}\label{prop:(T,w)-deviation_poly_to_find}
    Given an election instance $(\prof{B},\C)$ and an outcome $\vec{w}\in \C$, there exists an algorithm that finds all groups of voters $N'$ such that there exists an $S\subseteq \I$ with $N'$ having an $(S,\vec{w},\C)$-deviation, that runs in $O((\max_{t\in[m]}|D_{t}|)^{m}|\C|^{2}mn)$ time.
\end{proposition}

\begin{proof}{}
    Take $(\prof{B},\C)$ and outcome $\vec{w}\in \C$. Consider the following algorithm that operates in $|\C|$ rounds, assessing an outcome $\vec{w}\in \C$ in each round (with each outcome assessed once throughout): at each round for an outcome $\vec{w}\in \C$, iterate through all other outcomes $\vec{w}^{*}\neq \vec{w}\in \C$; for all sets of issues that $\vec{w}$ and $\vec{w}^{*}$ disagree on, we fix $S$ to be such a set (there are at most $(\max_{t\in[m]}|D_{t}|)^{m}$ of these sets where $\max_{t\in[m]}|D_{t}|$ is the maximal size of any issue's domain); in at most $mn$ steps, it can be checked if there is a set of voters 
    that agree with $\vec{w}^{*}$ on all issues in $S$ which verifies the existence of a 
    voter group $N'$ with an $(S,\vec{w},\C)$-deviation; keep a count of all such groups $N'$; if all outcomes have been assessed, terminate and return the number of groups identified that have a $(S,\vec{w},\C)$-deviation, otherwise, move to the next outcome. This algorithm takes $O((\max_{t\in[m]}|D_{t}|)^{m}|\C|^{2}mn)$ time to complete in the worst case. 
\end{proof}

\noindent
Given this result, we continue with the following computational result regarding the checking of whether $\cEJR$ is satisfied by an outcome of some election instance.

\begin{proposition}\label{prop:c-EJR_poly_to_find}
 Given an election instance $(\prof{B},\C)$ and an outcome $\vec{w}\in \C$, there exists an algorithm that decides in $O(2(\max_{t\in[m]}|D_{t}|)^{m}|\C|^{3}mn)$ time whether outcome $\vec{w}$ provides $\cEJR$. 
\end{proposition}

\begin{proof}{}
        From Proposition~\ref{prop:(T,w)-deviation_poly_to_find} we know that, given an outcome $\vec{w}$, we can find all groups with some $(S,\vec{w},\C)$-deviation for some $S\subseteq \I$ in $O(\max_{t\in[m]}|D_{t}|)^{m}|\C|^{2}mn)$ time. There can be at most $(\max_{t\in[m]}|D_{t}|)^{m}(|\C|-1)$ such groups (recall that $\max_{t\in[m]}|D_{t}|$ is the maximal size of any issue's domain). Then, for each group $N'$ with an $(S,\vec{w},\C)$-deviation, we can check their size in polynomial time and thus verify whether they are $T$-cohesive with $S\subseteq T$, and if so, we can check if there exists any voter $i\in N'$ with $\sat{i}{\vec{w}} > |T|$. 
\end{proof}

\noindent
This result is similarly negative in comparison to Proposition~\ref{prop:c-EJR_poly_to_find} as it is exponential in the number of issues $m$. Moving on, we mention that \citet{chandakARXIV2024sequential} have already shown that, in general, cohesiveness-EJR is not always satisfiable in their sequential decisions model. This negative result carries over to the unconstrained public-decision setting and thus, to the constrained setting as well. Although we shall, in the sections to follow, analyse the extent to which we can achieve positive results with cohesiveness-EJR in our constrained setting, this negative result motivates the study of the following weaker axiom---which is an adaptation of the JR axiom from committee elections---that can always be satisfied in the public-decision setting without constraints \citep{bulteauHPRT2021perpetualJR,freemanEtAlIJCAI2020variable,chandakARXIV2024sequential}.

\begin{definition}[$\cJR$]\label{def:cJR}
    Given an election instance $(\prof{B}, \C)$, an outcome $\vec{w}$ provides \emph{$\cJR$} if for every $T$-cohesive group of voters $N'\subseteq N$ for some $T\subseteq \I$ with an $(S,\vec{w},\C)$-deviation for some $S\subseteq T$ where $|S| = |T| = 1$, there exists a voter $i\in N'$ such that:
    $\sat{i}{\vec{w}} \geq 1$.
\end{definition}

\noindent
This is a relaxation of $\cEJR$ where the only voter groups that need be guaranteed any amount of satisfaction are those voter groups that are cohesive on, and have a feasible deviation with, a single issue. $\cJR$ aims to ensure that such voter groups have at least one voter that agrees with the outcome's decision on at least one issue.

\begin{example}\label{exm:cJR}
    Take the same election instance from Example~\ref{exm:cEJR} with three issues $\I = \{a_1, a_2, a_3\}$ with domains $D_{t} = \{d_{1},d_{2},d_{3}\}$ for all $t\in\{1,2,3\}$. Then take the constraint to be $\C  = \{(d_{1},d_{1},d_{2}),(d_{2},d_{2},d_{2}),(d_{2},d_{2},d_{3})\}$. Then, suppose there are three voters $N = \{1,2,3\}$ with ballots $\vec{b}_{1} = (d_{1},d_{1},d_{1})$, $\vec{b}_{2} = (d_{1},d_{1},d_{2})$ and $\vec{b}_{3} = (d_{3},d_{3},d_{3})$. Take outcome $\vec{w}=(d_{2},d_{2},d_{2})$. The voter group $\{1,2\}$ is a $T$-cohesive group for the set of issues $T = \{a_1, a_2\}$, both are not satisfied by any decisions by outcome $\vec{w}$ and they have an $(S,\vec{w},\C)$-deviation (with $S=T$) to the alternative outcome $(d_{1},d_{1},d_{2})$ but we have that $|S|=|T|=2>1$ and thus, $\cJR$ does not capture their complaint. In fact, the voter group $\{1,2\}$ does not have a feasible deviation of only a single issue for any of the outcomes in $\C$. So, this pair of voters are not witness to a violation of $\cJR$. On the other hand, voter~$3$ represents a $T$-cohesive group for $T=\{a_3\}$ and they have an $(S,\vec{w},\C)$-deviation of size $|S|=|T|=1$ to the outcome $(d_{2},d_{2},d_{3})$. Thus, outcome $\vec{w}=(d_{2},d_{2},d_{2})$ does not provide $\cJR$ but the outcome $\vec{w}=(d_{2},d_{2},d_{3})$ does.
\end{example}

\noindent
The above example demonstrates that $\cJR$ is a fairly tame demand. Unfortunately, when considering arbitrary constraints, even $\cJR$ cannot always be achieved. Note that this even holds for binary election instances.

\begin{proposition}\label{prop:cohJR_not_always_satisfiable}
There exists an election instance where no outcome provides $\cJR$.
\end{proposition}

\begin{proof}
Consider the binary election instance with issues $\I = \{a_1,a_2\}$ and a constraint $\C = \{(0,1),(0,0)\}$. Suppose that $N = \{1,2\}$, where $\vec{b}_{1} = (1,1)$ and $\vec{b}_{2} = (1,0)$. Note that for both outcomes $\vec{w}\in \C$, one voter will have satisfaction of $0$ while being a $T$-cohesive group with an $(S,\vec{w},\C)$-deviation for $|S| = |T| = 1$. As each voter is half of the population, they would have a deviation towards the alternative feasible outcome obtained by `flipping' issue $a_2$, which provides them greater satisfaction than the current one.
\end{proof}

\noindent
This result still holds if we restrict voters' ballots to be consistent with $\C$ (see  Proposition~\ref{prop:c-JR_may_not_exist_A1}'s counterexample).

Moving on, we propose a restriction to the class of constraints that we consider, given the negative results for the general case. To do so, we introduce notation for the \emph{fixed decisions} for a set of outcomes $C \subseteq \C $, which are the issues in $\I$ whose decisions are equivalent across all the outcomes in $C$. For a set of outcomes $C \subseteq \C$, we represent this as:
\[
\I_{\text{fix}}(C) = \{a_t\in \I\mid  \text{there exists some }d\in D_{t} \text{ such that } w_{t} = d  \text{ for all }\vec{w} \in C\}.
\]

\begin{definition}[No Fixed Decisions (NFD) property]
  A constraint $\C$ has the \emph{NFD} property if $\I_{\text{fix}}(\C) = \emptyset$ holds for $\C$.  
\end{definition}

\noindent
While the NFD property seems rather natural,
we argue that there are cases in which considering 
election instances where decisions that are fixed from the get-go may contribute to the satisfaction of voters and, specifically for our goal, these fixed decisions may aid in giving the voters their fair, proportional representation. 

Now, we show that with the NFD property, the $\cEJR$ axiom can always be satisfied, albeit only for election instances where either $|\C|=2$ (i.e., there are only two feasible outcomes) or $m\leq 3$ (i.e., the number of issues is lower than three). We begin with the (essentially) trivial case where the number of feasible outcomes is limited to two.

\begin{proposition}\label{prop:c-EJR_satisfiable_|C|=2_A1}
For election instances $(\prof{B}, \C)$ with $|\C| = 2$ where $\C$ has the $\text{NFD}$ property, $\cEJR$ can always be satisfied.
\end{proposition}

\begin{proof}{}
    Take some feasible outcome $\vec{w}\in \C$. Observe that when $|\C| = 2$, if property $\text{NFD}$ holds, then the two feasible outcomes differ on the decisions of all issues. Thus, it is only possible for $T$-cohesive groups with an $(S,\vec{w},\C)$-deviation for $|S| \leq |T| = m$ to have an allowable deviation from $\vec{w}$ to the only other feasible outcome. This means only the entire voter population has the potential to deviate. And if such deviation to $\vec{w}'$ exists, then outcome $\vec{w}'$ sufficiently represents the entire voter population.
\end{proof}

\noindent
Now we ask the following: can we guarantee $\cEJR$ when $m \leq 3$? We answer in the positive when we restrict ourselves to  binary election instances.

\begin{proposition}\label{prop:c-EJR_satisfiable_|I|<=3_A1}
   For binary election instances $(\prof{B}, \C)$ with $m \leq 3$ where the constraint $\C$ has the NFD property, $\cEJR$ can always be provided.
\end{proposition}

\begin{proof}{} The case for $m=1$ is trivially satisfied (since with only one issue, we get there only being two feasible outcomes and thus, any possible deviating group has to be of size~$n$) so we present the proof as two separate cases where the number of issues is either $m=2$ or $m=3$.

\myparagraph{\textbf{Case }$\boldsymbol{m = 2}$:}  Observe that for two binary issues (i.e., $m = 2$), there are $16$ non-empty constraints $\C$, but only $7$ of them satisfy the NFD property. Take one such $\C$ and a feasible outcome $\vec{w} = (d_{1},d_{2}) \in \C$ where $d_{1},d_{2} \in \{0,1\}$. 

We have to show that for every $T$-cohesive groups of voters $N'\subseteq N$ for some $T\subseteq \I$ with an $(S,\vec{w},\C)$-deviation for some non-empty $S\subseteq T$, there exists a voter $i\in N'$ such that $\sat{i}{\vec{w}} \geq |T|$. Since $|\I| = m = 2$, the agreement among voters (i.e., the set $T$) and their potential deviation (i.e., the set $S$) may concern at most two issues: namely, $|S|, |T| \in\{1,2\}$. 

First, consider $|T| = 1$. Since $|S| \leq |T|$ and $S \neq \emptyset$, we have $|S| = 1$ for any $T$-cohesive group (which is thus of size $|N'| \geq \nicefrac{n}{2}$) wishing to perform an $(S,\vec{w},\C)$-deviation from $\vec{w}$ to some other feasible outcome $\vec{w}'\in\C$. If there is a voter $i \in N'$ such that $\sat{i}{\vec{w}} \geq 1$, group $N'$ would be sufficiently satisfied and $\cEJR$ would be ensured. Otherwise, we have that $\sat{i}{\vec{w}} = 0$ for all $i \in N'$ and they are unanimous on both issues, i.e., $\vec{b}_i = (1-d_{1}, 1-d_{2})$ for all $i \in N'$. There are just two possible outcomes that differ from $\vec{w}$ on only one issue. If neither outcome is in $\C$, then no feasible deviation is possible for $N'$ and we are done. Otherwise, assume without loss of generality that $\vec{w}' = (1-d_{1}, d_{2}) \in \C$. Now, if there is a voter $i \in N \setminus N'$ such that $\sat{i}{\vec{w}'} \geq 1$, then we are done (as the group $N \setminus N'$ would be sufficiently satisfied if it were $T$-cohesive for $|T| = 1$). Else, it means that all voters $j \in N \setminus N'$ are unanimous on ballot $\vec{b}_j = (d_{1}, 1-d_{2})$. But then, since $\C$ satisfies property $\text{NFD}$, there exists some outcome $\vec{w}''\in\C$ such that $\vec{w}''_2 = 1-d_{2}$. Then, $\sat{i}{\vec{w}''} \geq 1$ for all $i \in N$ and no deviation is possible.
    
Now, consider $|T| = m = 2$. For a voter group $N'$ to be $T$-cohesive, they must unanimously agree on both issues and their size must be $|N'| = n$. In order for the group $N'$ to have an ($S,\vec{w}, \C$)-deviation for $|S| \leq |T|$, it must be the case that $\sat{i}{\vec{w}} = 0$ for all $i \in N'$. By property $\text{NFD}$, there must be some outcome $\vec{w}' \neq \vec{w} \in \C$, and thus $\sat{i}{\vec{w}'} \geq 1$ for all $i \in N$.

\myparagraph{\textbf{Case }$\boldsymbol{m = 3}$:} Let $(\prof{B}, \C)$ be an election instance satisfying the conditions in the statement. We now reason on the existence of possible $T$-cohesive groups that are a witness to the violation of $\cEJR$, for each possible size $1 \leq |T| \leq 3$ of the set $T$. 

For $|T|=1$, suppose by contradiction that for all $\vec{w} \in \C$, there is some voter group $N'$ such that $|N'|\geq \nicefrac{n}{3}$ and each voter in $N'$ has satisfaction of~$0$. Thus, for all voters $i \in N'$ we have $\vec{b}_i= \vec{\bar w}$. Moreover, for a $T$-cohesive group with an ($S,\vec{w}$)-deviation for $|S| = |T| = 1$ to be possible, there has to exist a $\vec{w}'\in \C$ whose decisions differ from $\vec{w}$ in exactly one issue, i.e., $\Agr{\vec{w}}{\vec{w}'} = 2$. To fit all these disjoint $T$-cohesive groups for $|T| = 1$, one for each outcome in $\C$, it must be that $n\geq |\C|\cdot \nicefrac{n}{3}$, hence $|\C|\leq 3$ must hold. If $|\C| = 1$, the NFD property cannot be met. If $|\C| = 2$, then the two feasible outcomes cannot differ in the decision of only one issue while also satisfying the NFD property. For $|\C|=3$, to get a $T$-cohesive voter group with an ($S,\vec{w}$)-deviation for $|S| = |T| = 1$ at every $\vec{w}\in\C$, the three feasible outcomes must differ by at most one decision, contradicting the NFD property.

For $|T| = 2$, we only consider ($S,\vec{w}$)-deviations from a $T$-cohesive group $N'$ with $|S| \in \{1,2\}$. Consider the case of $|S|=1$. W.l.o.g., assume that $\vec{w} = (0,0,0)$ and that there exists a $T$-cohesive group $N'$ (where $|N'| \geq n\cdot\nicefrac{2}{3}$) with every voter having satisfaction $\sat{i}{\vec{w}} < 2$, with an ($S,\vec{w}$)-deviation towards some outcome, e.g., $\vec{w}' = (1,0,0)$. So $\sat{i}{\vec{w}'} \geq 1$ holds for at least one voter $i\in N'$. If $N'$ has no ($S,\vec{w}'$)-deviation, then there are not a witness to a violation of $\cEJR$. Otherwise, suppose no voter in $N'$ has satisfaction of $2$ (so all voters in $N'$ have a ballot $(1,1,1)$) and the group $N'$ has an ($S,\vec{w}'$)-deviation, to outcome $\vec{w}'' = (1,1,0)$. Then all voters in $N'$ have satisfaction of at least $2$ with outcome $\vec{w}''$. Suppose that the remaining $\nicefrac{n}{3}$ of the voters $N\setminus N'$ have an ($S,\vec{w}''$)-deviation (so each of these voters derives zero satisfaction from outcome $\vec{w}''$ and all disagree with group $N'$ on the first two issues $a_1$ and $a_2$). This must be a deviation of size $|S|=1$ for the $\nicefrac{n}{3}$ of the voters $N\setminus N'$ to demand it. This deviation can only be to outcome $(1,1,1)$, $(0,1,0)$ or $(1,0,0)$. By the NFD property, we know that one of the outcomes $\{(0,0,1), (0,1,1), (1,0,1),(1,1,1)\}$ must be in $\C$. See that for each of the outcomes $(1,1,1)$, $(1,0,1)$ and $(0,1,1)$, some voter in $N'$ gets at least satisfaction of $2$ while all voters in $N\setminus N'$ get at least $1$ in satisfaction, thus, $\cEJR$ is satisfied. Now, out of these outcomes $\{(0,0,1), (0,1,1), (1,0,1),(1,1,1)\}$, if only outcome $\vec{w}''' = (0,0,1)$ is in $\C$, then group $N'$ has no ($S,\vec{w}'''$)-deviations of size $|S| = 1$ as this is only possible to one of $(0,1,1)$ or $(1,0,1)$. So in this case too, $\cEJR$ is satisfied. 

Now we look at the case for $|S| = 2$. W.l.o.g., consider the outcome $\vec{w} = (0,0,0)$ and assume that there exists a $T$-cohesive group $N'$ (where $|N'| \geq n\cdot\nicefrac{2}{3}$) with an ($S,\vec{w}$)-deviation towards outcome, e.g., $\vec{w}' = (1,1,0)$. Thus, there is some voter $i$ in $N'$ with satisfaction $\sat{i}{\vec{w}'}\geq2$. At this point, the only possible further ($S,\vec{w}$)-deviation could arise for $|S| = 1$ in case there are $\nicefrac{n}{3}$ voters in $N\setminus{N'}$ each have a satisfaction of 0 for $\vec{w}'$, i.e., each has the ballot $(0,0,1)$ and either one of the outcomes in $\{(1,0,0),(0,1,0),(1,1,1)\}$ is in $\C$. Now take instead that $\sat{i}{\vec{w}'} = 2$ and consider two cases where either voter $i$ agrees or disagrees  with the voters in $N\setminus{N'}$ on the decision of issue $a_3$. First, assume that voter $i\in N'$ \emph{agrees} with the voters in $N\setminus{N'}$ on issue $a_3$ (so voter $i$ had the ballot $\vec{b}_{i} = (1,1,1)$). Then if either $(0,1,1)\in\C$ or $(1,1,1)\in\C$ holds, we have that $\cEJR$ is provided. And if $(0,0,1)\in\C$ holds, then voters in $N\setminus{N'}$ are entirely satisfied and the voters in $N'$ may only have an ($S,\vec{w}$)-deviation for $|S| \leq |T| = 2$ if either $(0,1,1)\in\C$ or $(1,1,1)\in\C$ holds (as they only `flip' issues they disagree with), which means that $\cEJR$ is provided. In the second case, assume that voter $i\in N'$ \emph{disagrees} with the voters in $N\setminus{N'}$ on issue $a_3$ and so, voter $i$ had the ballot $\vec{b}_{i} = (1,1,0)$. This means that $\sat{i}{\vec{w}'} = 3$ holds, hence, any outcome that the voters in $N\setminus{N'}$ propose given they have an ($S,\vec{w}$)-deviation for $|S| = 1$, would be one that provides $\cEJR$. 

Finally, a $T$-cohesive group for $|T| = 3$ implies a unanimous profile; if there exists an allowable ($S,\vec{w}$)-deviation for $|S| \leq |T| = 3$, then the outcome in $\C$ maximising the sum of agreement with the profile provides $\cEJR$.
\end{proof}

\noindent
We leave it open whether the above result holds if we do not restrict our view to binary election instances. Unfortunately, we provide an example showing that $\cJR$ cannot be guaranteed in general even when the NFD property holds for binary election instances.

\begin{proposition}\label{prop:c-JR_may_not_exist_A1}
    There exists an election instance $(\prof{B}, \C)$ where $m=8$ and the constraint $\C$ has the NFD property but no $\cJR$ outcome exists.
\end{proposition}

\begin{proof}
    Suppose there is a binary election instance with a constraint $\C = \{\vec{w}_1, \vec{w}_2, \vec{w}_3, \vec{w}_4\}$ for $m=8$ such that $\vec{w}_1 = (0,0,0,\ldots,0)$, $\vec{w}_2 = (0,0,1,\ldots,1)$, $\vec{w}_3 = (1,1,1,\ldots,1)$ and $\vec{w}_4 = (1,1,0,\ldots,0)$. Consider now a profile of four voters where $\vec{b}_i = \vec{w}_i$. Given that $m=8$, note that for every outcome $\vec{w} \in \C$, there exists some voter that deserves $2$ in satisfaction by being $T$-cohesive for $|T| = 2$ with an $(S,\vec{w},\C)$-deviation but with zero in satisfaction. And by $\cJR$, such a voter would be entitled to at least $1$ in satisfaction, so there is no outcome in $\C$ that provides $\cJR$.
\end{proof}

\noindent
We now turn our attention towards a weakening of $\cEJR$ that takes inspiration from the `up-to-one' relaxation of EJR studied in the context of participatory budgeting \citep{pierczynskiSP2021propPB,reyM2023PBsurvey}.

\begin{definition}[$\cEJR$-1]\label{def:c-EJR-1}
    Given an election $(\prof{B}, \C)$, an outcome $\vec{w}$ provides \emph{ $\cEJR$-1} if for every $T$-cohesive group of voters $N'\subseteq N$ for some $T\subseteq \I$ with an $(S,\vec{w},\C)$-deviation for some non-empty $S\subseteq T$, there exists a voter $i\in N'$ such that:
    \[
    \sat{i}{\vec{w}} \geq |T| -1
    \]
\end{definition}

\noindent
Intuitively, this axiom states that a fair outcome would be one where all groups that are can demand a certain level of satisfaction as defined by $\cEJR$, are a single decision away from obtaining their deserved satisfaction. As $\cEJR$ implies $\cEJR$-1, the results of Propositions \ref{prop:c-EJR_satisfiable_|C|=2_A1} and \ref{prop:c-EJR_satisfiable_|I|<=3_A1} immediately apply to $\cEJR$-1.

\begin{corollary}\label{cor:c-EJR-1_satisfiable_|C|=2_A1}
    For binary election instances $(\prof{B}, \C)$ with $|\C| = 2$ where the constraint $\C$ has the NFD property, $\cEJR$-1 can always be provided.
\end{corollary}

\begin{corollary}\label{cor:c-EJR-1_satisfiable_|I|=3_A1}
    For binary election instances $(\prof{B}, \C)$ with $m \leq 3$ where the constraint $\C$ has the NFD property, $\cEJR$-1 can always be provided.
\end{corollary}

\noindent
Note that for the computational result for $\cEJR$ in Proposition~\ref{prop:c-EJR_poly_to_find}, a simple alteration of the proof given for Proposition~\ref{prop:c-EJR_poly_to_find} (replacing the value $|T|$ with $|T|-1$ in the final satisfaction check) yields a corresponding computational result for $\cEJR$-1 that is similarly (negatively) impacted by the number of issues in the election instance. 

\begin{proposition}\label{prop:c-EJR-1_poly_to_find}
 Given an election instance $(\prof{B},\C)$ and an outcome $\vec{w}\in \C$, there exists an algorithm that decides in $O(2(\max_{t\in[m]}|D_{t}|)^{m}|\C|^{3}mn)$ time whether outcome $\vec{w}$ provides $\cEJR$-1. 
\end{proposition}

\noindent
Next, for the result that shows that $\cEJR$ can be provided when $m=2$ given that NFD holds (see Proposition~\ref{prop:c-EJR_satisfiable_|I|<=3_A1}), we can show something stronger for $\cEJR$-1 by dropping the assumption that the NFD property holds.

\begin{proposition}\label{prop:c-EJR-1_always_satisfiable_m=2}
        For election instances $(\prof{B}, \C)$ with $m = 2$, $\cEJR$-1 can always be provided.
\end{proposition}

\begin{proof}{}
Consider an election over two issues, where a $T$-cohesive group of voters has an $(S,\vec{w},\C)$-deviation for some outcome $\vec{w}$, as per Definition~\ref{def:c-EJR-1}. Observe that, when $m = 2$, $(S,\vec{w},\C)$-deviation are only possible for $|S|\in\{1,2\}$. 
Take a $T$-cohesive group $N'$ for $|T| = 1$ with an $(S,\vec{w},\C)$-deviation from $\vec{w}$ to some other feasible outcome $\vec{w}'\in\C$. Even if $\sat{i}{\vec{w}} = 0$ for every voter $i \in N'$, we have $\sat{i}{\vec{w}} \geq |T| - 1 = 1-1 = 0$, and thus $\cEJR$-1 is satisfied.
Take now a $T$-cohesive group $N'$ for $|T| = 2$: for them to deviate, it must be the case that $N' = N$, and $\sat{i}{\vec{w}} = 0$ for all $i \in N$. If they have an $(S,\vec{w},\C)$-deviation for $|S| = |T| = 2$, the outcome $\vec{w'}$ they wish to deviate to must increase the satisfaction of each voter by at least $1$, which thus satisfies $\sat{i}{\vec{w}} \geq |T| - 1 = 2-1 = 1$, and thus $\cEJR$-1.
\end{proof}

\noindent
Can we show that an outcome providing $\cEJR$-1 always exists when there are more than three issues, unlike for $\cEJR$? Unfortunately, this is not the case, even assuming property $\text{NFD}$, as the same counterexample used to prove Proposition~\ref{prop:c-JR_may_not_exist_A1} yields the following (so also for binary election instances).

\begin{proposition}\label{prop:c-EJR-1_may_not_exist_A1}
     There exists an election instance $(\prof{B}, \C)$ where $m = 8$ and the constraint $\C$ has the NFD property but there exists no outcome that provides $\cEJR$-1.
\end{proposition}

\noindent
We demonstrate that the challenge of satisfying $\cEJR$-1 lies in the constraints. In fact, we show that in the setting without constraints, it is always possible to find an outcome that provides $\cEJR$-1. To do so, we define the constrained version of MES that has been studied for the public-decision setting without constraints. Our adaptation allows for the prices associated with fixing the outcome's decisions on issues to vary. This contrasts with the unconstrained MES that fixes the prices of every issue's decision to $n$ from the onset. This pricing is determined by a particular pricing type $\lambda:\I\times (D_{1}\cup\ldots\cup D_{m})\rightarrow \mathbb{R}_{\geq 0}$ which maps an issue-decision pair to a non-negative price. 

\begin{definition}[$\cMES$]
The rule runs for at most $m$ rounds. Each voter has a budget of $m$. In every round, for every undecided issue $a_t$ in a partial outcome $\vec{w}^{*}$, we identify those issue-decision pairs $(a_t,d)$ where fixing some decision $d\in D_{t}$ on issue $a_t$ allows for a feasible outcome to be returned in future rounds. Otherwise, for every such pair $(a_t,d)$, we calculate the minimum value for $\rho_{(a_t,d)}$  such that if each voter in $N(a_t,d)$ were to pay either $\rho_{(a_t,d)}$ or the remainder of their budget, then these voters could afford to pay the price $\lambda(a_t,d)$ (determined by the pricing type $\lambda$). If there exists no such value for $\rho_{(a_t,d)}$, then we say that the issue-decision pair $(a_t,d)$ is \emph{not affordable} in round, and if in a round, there are no affordable issue-decision pairs, the rule stops. 
Otherwise, we update $\vec{w}^{*}$ by setting decision $d$ on issue $a_t$ for the pair $(a_t,d)$ with a minimal value $\rho_{(a_t,d)}$ (breaking ties arbitrarily, if necessary) and have each voter in $N(a_t,d)$ either paying $\rho_{(a_t,d)}$, or the rest of their budget. Note that $\cMES$ may terminate with not all issues being decided and we assume that all undecided issues are decided arbitrarily. 
\end{definition}

\noindent
A natural candidate for a pricing type is the standard pricing of unconstrained MES where the price for every issue-decision pair $(a_t,d)$ is set to $\lambda(a_t,d) = n$. We refer this pricing as \emph{unit} pricing $\lambda_{\text{unit}}$.

\begin{example}
    Suppose there are three issues $\I = \{a_1, a_2, a_3\}$ with domains $D_{t} = \{d_{1},d_{2},d_{3}\}$ for all $t\in\{1,2,3\}$. Then take the constraint to be $\C  = \{(d_{1},d_{2},d_{2}),(d_{1},d_{3},d_{3}),(d_{3},d_{2},d_{3})\}$. Then, suppose there are three voters $N = \{1,2,3\}$ with ballots $\vec{b}_{1} = (d_{1},d_{2},d_{1})$, $\vec{b}_{2} = (d_{1},d_{2},d_{2})$ and $\vec{b}_{3} = (d_{3},d_{3},d_{2})$. We consider the execution of $\cMES$ with a unit pricing $\lambda_{\text{unit}}$ so we have that each issue costs $n = 3$ and each voter has a budget of $m= 3$. In every round, setting any decision would cost voter~$3$ has a price of $3$ as no other voter agrees with them on any decision. Now, take the first round. Voter group $\{1,2\}$ would set either issue~$a_1$ to decision $d_{1}$, or issue~$a_2$ to decision $d_{2}$, at a cost of $\nicefrac{3}{2}$ to each of them. Suppose it was the latter and thus, outcome $(d_{1},d_{3},d_{3})$ is no longer feasible. In the second round, voter~$3$ setting either issue~$a_1$ or issue~$a_3$ to decision $d_{3}$ would be feasible with respect to the constraint, however, the rule favours voter group $\{1,2\}$ setting issue $a_1$ to $d_{1}$ as they (again) pay $\nicefrac{3}{2}$ each. Thus, outcome $(d_{3},d_{2},d_{3})$ is also not feasible. As voter~$3$ does not approve of a decision that remains feasible and the voters in $\{1,2\}$ have spent the of their budgets, the rule terminates with issue~$a_3$ not being set and the partial outcome $(d_{1},d_{2},\_)$ is returned. However, note that with $(d_{1},d_{2},d_{2})$ there is a feasible completion of this partial outcome in $\C$.
\end{example}

\noindent
Next, we show that in unconstrained elections, we can use $\cMES$ to provides fairness guarantees that are `close to' $\cEJR$-1. We first define this approximate version of $\cEJR$.

\begin{definition}[$\alpha\text{-}{\cEJR}\text{-}\beta$]
    Given an election $(\prof{B}, \C)$, some $\alpha\in (0,1]$ and some positive integer $\beta$, an outcome $\vec{w}$ provides \emph{$\alpha\text{-}{\cEJR}\text{-}\beta$} if for every $T$-cohesive group of voters $N'\subseteq N$ for some $T\subseteq \I$ with an $(S,\vec{w},\C)$-deviation for some non-empty $S\subseteq T$, there exists a voter $i\in N'$ such that:
    \[
    \sat{i}{\vec{w}} \geq \alpha\cdot|T| - \beta.
    \]
\end{definition}

\noindent
The multiplicative and additive factors in this axiom allow us to measure how well rules satisfy this notion even if they fall short of providing the ideal representation. Now, here is our result for $\cMES$.

\begin{proposition}\label{prop:MES_satisfies_EJR-1}
    For election instances, when $\C = \{0,1\}^{m}$, $\cMES$ with unit pricing $\lambda_{\text{unit}}$ satisfies $\nicefrac{2}{(\max_{t\in[m]}|D_{t}|)}\text{-}{\cEJR}\text{-}1$.
\end{proposition}

\begin{proof}{}
Take an outcome $\vec{w}$ returned by $\cMES$ with unit pricing $\lambda_{\text{unit}}$ and consider a $T$-cohesive group of voters $N'$. Let us assume that for every voter $i\in N'$, it holds that $\sat{i}{\vec{w}} < (\nicefrac{2}{\max_{t\in[m]}|D_{t}|})\cdot|T|-1$ and then set $\ell = (\nicefrac{2}{\max_{t\in[m]}|D_{t}|})\cdot|T|-1$. So to conclude the run of $\cMES$, each voter in $N'$ paid for at most $\ell-1 = (\nicefrac{2}{\max_{t\in[m]}|D_{t}|})\cdot|T|-2$ issues. 

    Now, assume that the voters in $N'$ paid at most $\nicefrac{m}{(\ell + 1)}$ for any decision on an issue. We know that each voter has at least the following funds remaining at that moment:
\begin{equation*}
m - (\ell - 1)\frac{m}{\ell + 1}  = \frac{2m}{\ell + 1} = \frac{(\max_{t\in[m]}|D_{t}|)m}{|T|} \geq \frac{(\max_{t\in[m]}|D_{t}|)n}{|N'|}.  
\end{equation*}

 The last step follows from the group $N'$ being $T$-cohesive (and thus, $|N'| \geq |T|\cdot\nicefrac{n}{m}$). So now we know that the voters in $N'$ hold at least $(\max_{t\in[m]}|D_{t}|)n$ in funds at the end of $\cMES$'s run. Thus, we know that at least $\max_{t\in[m]}|D_{t}|$ issues have not been funded and for at least one of these issues, at least an $\nicefrac{1}{(\max_{t\in[m]}|D_{t}|)}$-fraction of $N'$ agree on the decision of this issue (as the election instance has at most $\max_{t\in[m]}|D_{t}|$ alternatives for any issue) and they hold enough funds to pay its price of $n$ (given by the unit pricing of $\cMES$), hence, we have a contradiction to $\cMES$ terminating. 

Now, assume that some voter $i$ in $N'$ paid more than $\nicefrac{m}{(\ell + 1)}$ for a decision on an issue. Since we know that at the end of $\cMES$'s execution, each voter in $N'$ paid for at most $\ell-1 = (\nicefrac{2}{\max_{t\in[m]}|D_{t}|})\cdot|T|-2$ issues, it must be that in the round~$r$ that voter~$i$ paid more than $\nicefrac{m}{(\ell + 1)}$ for an issue's decision, the voters in $N'$ collectively held at least $(\max_{t\in[m]}|D_{t}|)n$ in funds. But, there are at least $\max_{t\in[m]}|D_{t}|$ issues in $\I$ that were not funded, so there exists some issue that could have been paid for in round~$r$ where voters each pay $\nicefrac{m}{(\ell + 1)}$. This contradicts the fact that voter $i$ paid more than $\nicefrac{m}{(\ell + 1)}$ in round $r$. 
So, we have that this group of voters $N'$ cannot exist and that $\cMES$ satisfies $\nicefrac{2}{(\max_{t\in[m]}|D_{t}|)}\text{-}{\cEJR}\text{-}1$.
\end{proof}

\noindent
Note that for unconstrained binary elections, the above result states that $\cMES$ with unit pricing $\lambda_{\text{unit}}$ provides $\cEJR$-1.\footnote{Note that \citet{skowronGoreckiAAAI2022public} gave a similar result for $\cMES$ with unit pricing $\lambda_{\text{unit}}$ by showing that it satisfies an axiom based on the agreement-EJR notion in the unconstrained binary setting (see Theorem~$2$ in their paper).} However, the proportionality guarantee weakens when more than two issues are considered for issues in the election. Recall that results by \citet{chandakARXIV2024sequential} showed that we cannot get $\cEJR$ even in the unconstrained setting. Thus, the above positive result shows the promise of $\cMES$ in providing strong representation guarantees when we have constraints.

\section{Justified Representation with Agreement}\label{sec:justified_agreement}

Given the mostly negative results regarding the (stronger) cohesiveness-EJR notion, we move on to justified representation based on agreement, which leads to weaker requirements in the unconstrained setting. 
First, we formalise agreement-based EJR with the following axiom.

\begin{definition}[$\agrEJR$]\label{def:agrEJR}
    Given an election $(\prof{B}, \C)$, an outcome $\vec{w}$ provides \emph{$\agrEJR$} if for every $T$-agreeing group of voters $N'\subseteq N$ for some $T\subseteq \I$ with an $(S,\vec{w},\C)$-deviation for some $S\subseteq T$ with $|S|\leq |T|\cdot\nicefrac{|N'|}{n}$, there exists a voter $i\in N'$ such that:
    \[\sat{i}{\vec{w}} \geq |T|\cdot\frac{|N'|}{n}.\]
\end{definition}

\noindent
Intuitively, $\agrEJR$ states that a (member of a) $T$-agreeing voter group deserves at least $|T|\cdot\nicefrac{|N'|}{n}$ in satisfaction from an outcome $\vec{w}$. However, if this is not the case, then fairness is only violated if said voter group can find a suitable set $S\subseteq T$ of issues that they can change towards deviating to a different, feasible outcome. Note a size requirement for the set $S$ has been set to $|S|\leq |T|\cdot\nicefrac{|N'|}{n}$. This is made to prevent scenarios where a voter group has an $(S,\vec{w},\C)$-deviation that is `too large', e.g., a voter group that agrees on all issues but constitutes $50\%$ of the population, cannot claim a violation of  $\agrEJR$ via an $(S,\vec{w},\C)$-deviation where $S = \I$. 

\begin{example}\label{exm:agrEJR}
    Suppose there are three issues $\I = \{a_1, a_2, a_3\}$ with domains $D_{t} = \{d_{1},d_{2},d_{3}\}$ for all $t\in\{1,2,3\}$. Then take the constraint to be $\C  = \{(d_{1},d_{2},d_{3}),(d_{2},d_{2},d_{3})\}$. Then, suppose there are three voters $N = \{1,2,3\}$ with ballots $\vec{b}_{1} = (d_{1},d_{1},d_{1})$, $\vec{b}_{2} = (d_{1},d_{1},d_{2})$ and $\vec{b}_{3} = (d_{3},d_{3},d_{3})$. Take the outcome to be $\vec{w} = (d_{2},d_{2},d_{3})$. Observe that voter group $
    \{1,2\}$ agree on the first two issues and they represent a $\nicefrac{2}{3}$-rd fraction of the population. Since this group has an $(S,\vec{w},\C)$-deviation for $S = \{1\}$ that is of the appropriate size (at most $2\cdot\nicefrac{2}{3}$), we see that $\vec{w} = (d_{2},d_{2},d_{3})$ violates $\agrEJR$.
    On the other hand, the outcome $\vec{w} = (d_{1},d_{2},d_{3})$ provides $\agrEJR$.
\end{example}

\noindent
Unfortunately, we find that $\agrEJR$ is not always satisfiable in general. This follows from the counterexample of Proposition~\ref{prop:c-JR_may_not_exist_A1}, as each voter requires at least $1$ in satisfaction for to $\agrEJR$ to be satisfied. 

\begin{corollary}\label{prop:agrEJR_not_always_satisfiable_NFD}
    There exists an election instance where no outcome provides $\agrEJR$ (even when $\C$ satisfies NFD).
\end{corollary}

\noindent
We now introduce a particular class of constraints that allows us to precisely define how restrictive, and thus how costly, the fixing of a particular issue-decision pair is. Similarly to work by \citet{reyEtAlKR2020general,reyEtAlSCW2023general}, we consider constraints $\C$ that can be equivalently expressed as a set of implications $\textit{Imp}_{\C}$, where each implication in $\textit{Imp}_{\C}$ is a propositional formula with the following form where $\ell_{(a_x,d_{x})}$ is a literal associated with the issue-decision pair $(a_x,d_{x})$: $\ell_{(a_x,d_{x})} \rightarrow \ell_{(a_y,d_{y})}$.
This class of constraints allows us, for instance, to express simple dependencies and conflicts such as `selecting $x$ means that we must select $y$' and `selecting $x$ means that $y$ cannot be selected', respectively. These constraints correspond to \emph{propositional logic formulas} in \emph{2CNF}. 

\begin{example}\label{exm:constraint_implications}
Take a set of issues $\I = \{a,b,c,d,e\}$ for a binary election instance. An example of an implication set is
$\textit{Imp}_{\C} = \{(a,1) \rightarrow (b,1),  (c,1) \rightarrow (e,0), (d,1) \rightarrow (e,0)\}$. Here, accepting issue $a$ means that issue $b$ must also be accepted while accepting either issues $c$ or $d$ requires the rejection of issue $e$.
\end{example}

\noindent
Given a (possibly partial) outcome $\vec{w}\in\C$ and the set $\textit{Imp}_{\C}$, we construct a directed \emph{outcome implication graph} $G_{\vec{w}} = \langle H,E\rangle$ where $H = \bigcup_{a_{t}\in \I}\{(a_{t},d)\mid d\in D_{t}\}$ 
as follows: 
\begin{enumerate}
    \item Add the edge $((a_x,d_{x}), (a_y,d_{y}))$ to $E$ if $\ell_{(a_x,d_{x})} \rightarrow \ell_{(a_y,d_{y})} \in\textit{Imp}_{\C}$  and $w_{y} \neq d_{y}$; 
    \item If there exists an implication  $\ell_{(a_x,d_{x})} \rightarrow \ell_{(a_y,d_{y})} \in\textit{Imp}_{\C}$ while $w_{x} = d_{x}$ holds, then add the edge $((a_y,d_{y}^{*}), (a_x,d_{x}^{*}))$ for all $d_{y}^{*}\neq d_{y}\in D_{y},d_{x}^{*}\neq d_{x}\in D_{x}$ to $E$. 
\end{enumerate}

Given such a graph $G_{\vec{w}}$ for an outcome $\vec{w}$, we use $G_{\vec{w}}(a_x,d_{x})$ to denote the set of all vertices that belong to some path in $G_{\vec{w}}$ having vertex $(a_x,d_{x})$ as the source. Note that $G_{\vec{w}}(a_x,d_{x})$ excludes $(a_x,d_{x})$.

\begin{example}\label{exm:constraint_graphs}
Consider a binary election instance and take a set of issues $\I = \{a_{1},a_{2},a_{3},a_{4}\}$ and the implication set $\textit{Imp}_{\C} = \{(a_{1},1) \rightarrow (a_{2},1), (a_{1},1) \rightarrow (a_{3},1), (a_{2},1) \rightarrow (a_{4},1)\}$ of some constraint $\C$. Consider the outcome implication graph for $\vec{w}_{1} = (0,0,0,0)$ (vertices with no adjacent edges are omitted for readability):

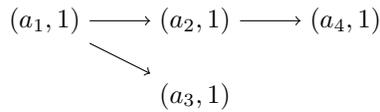
\begin{center}
		\begin{tikzpicture}[anchor=center]
			\node[name = A1] at (-1, 2) {$(a_{1},1)$};
            \node[name = C1] at (1, 1) {$(a_{3},1)$};
            \node[name = B1] at (1, 2) {$(a_{2},1)$};
            \node[name = D1] at (3, 2) {$(a_{4},1)$};

			\path[->] (A1) edge (B1);
            \path[->] (A1) edge (C1);
			\path[->] (B1) edge (D1);
		\end{tikzpicture}
\end{center}

Then, we have $G_{\vec{w}_{1}}(a_1,1) = \{(a_2,1), (a_3,1), (a_4,1)\}$ and therefore $|G_{\vec{w}_{1}}(a_1,1)| = 3$. Now, consider the outcome implication graph for $\vec{w}_{2} = (0,1,0,0)$:

\begin{center}
		\begin{tikzpicture}[anchor=center]
			\node[name = A1] at (-1, 2) {$(a_{1},1)$};
            \node[name = C1] at (1, 1) {$(a_{3},1)$};
            \node[name = B1] at (1, 2) {$(a_{2},0)$};
            \node[name = D1] at (3, 2) {$(a_{4},0)$};

            \path[->] (A1) edge (C1);
			\path[<-] (B1) edge (D1);
		\end{tikzpicture}
\end{center}

Here, we have that $G_{\vec{w}_{2}}(a_4,0) = \{(a_2,0)\}$ but note that there is no edge from $(a_2,0)$ to $(a_1,0)$ as the latter already holds in the outcome.
\end{example}

\noindent
Thus, for an issue-decision pair $(a_x, d_x)$, we can count the number of affected issues in setting a decision $d_{x}$ for the issue $a_x$. This leads us to the following class of constraints.

\begin{definition}[$k$-restrictive constraints]
Take some constraint $\C$ expressible as a set of implications $\textit{Imp}_{\C}$. For some positive integer $k\geq 2$, we say that $\C$ is \emph{$k$-restrictive} if for every outcome $\vec{w}\in\C$, it holds that:
\[
\max\bigg\{|G_{\vec{w}}(a_x,d_{x})| \biggm|  (a_x,d_{x})\in \bigcup_{a_{t}\in \I}\{(a_{t},d)\mid d\in D_{t}\}\bigg\} = k-1
\]
where $G_{\vec{w}}$ is the outcome implication graph constructed for outcome $\vec{w}$ and the implication set $\textit{Imp}_{\C}$.
\end{definition}

\noindent
Intuitively, with a $k$-restrictive constraint, if one were to fix/change an outcome $\vec{w}$'s decision for one issue, this would require fixing/changing $\vec{w}$'s decisions on at most $k-1$ other issues. 
In the remainder of the paper, when we refer to a $k$-restrictive constraint $\C$, we assume that $\C$ is expressible using an implication set $\textit{Imp}_{\C}$. 
%
We also mention in passing that the computational complexity of checking, for some constraint $\C$, whether there exists a set of implications $\textit{Imp}_{\C}$ that is equivalent to $\C$,  
is a problem that corresponds to \emph{Inverse Satisfiablility}, which has been shown to be polynomial for formulas in 2CNF \citep{KavvadiasS98inverseSAT}. 


\noindent
We now present an approximate variant of the agreement-EJR notion for public decisions with constraints in a similar manner as we did with $\cEJR$.

\begin{definition}[$\alpha\text{-}{\agrEJR}\text{-}\beta$]\label{def:alpha-beta-c-PROP}
    Given an election $(\prof{B}, \C)$, some $\alpha\in (0,1]$ and some positive integer $\beta$, an outcome $\vec{w}$ provides \emph{$\alpha\text{-}{\agrEJR}\text{-}\beta$} if for every $T$-agreeing group of voters $N'\subseteq N$ for some $T\subseteq \I$ with an $(S,\vec{w},\C)$-deviation for some non-empty $S\subseteq T$ with $|S|\leq |T|\cdot\nicefrac{|N'|}{n}$, there exists a voter $i\in N'$ such that:
    \[\sat{i}{\vec{w}} \geq \alpha\cdot|T|\cdot\frac{|N'|}{n} - \beta.\]
\end{definition}

\noindent
Note again that we place a size requirement on the set $S$ on which a group has an $(S,\vec{w},\C)$-deviation so that we rule out cases such as a single voter only having an $(S,\vec{w},\C)$-deviation for $S=\I$ while not intuitively being entitled to that much representation. 
Note that for readability, when we have $\alpha = 1$ or $\beta = 0$, we omit them from the notation when referring to $\alpha\text{-}{\agrEJR}\text{-}\beta$.

\begin{example}
 Suppose there is a binary election instance with four issues $\I = \{a_1, a_2, a_3, a_4\}$ and consider a constraint $\C = \{(1,1,0,0),(1,1,1,0)\}$. Then suppose there are two voters with ballots $\vec{b}_{1} = (1,1,1,1)$ and  $\vec{b}_{2} = (0,0,0,0)$ so each voter deserves at least $2$ in satisfaction according to the agreement-EJR notion. See that outcome $\vec{w} = (1,1,0,0)$ provides $\agrEJR$ while the outcome $\vec{w}' = (1,1,1,0)$ only provides $\nicefrac{1}{2}\text{-}{\agrEJR}$ as voter $2$ only obtains $1$ in satisfaction whilst having a sufficiently small $(S,\vec{w},\C)$-deviation for the issue $a_{3}$ (deviating to outcome $\vec{w}$). 
\end{example}

\noindent
We now analyse $\cMES$ with respect to this axiom for the class of $k$-restrictive constraints. We say that for $\cMES$, the price for an issue-decision pair $(a_x,d)$ given a partial outcome $\vec{w}^{*}$ is $\lambda(a_x,d) = n\cdot (|G_{\vec{w}^{*}}(a_x,d)|+1)$ and we refer to this as a \emph{fixed} pricing $\lambda_{\text{fix}}$.  Then we can show the following for binary election instances.

\begin{theorem}\label{thm:c-MES-fix_satisfies_1/k-agrEJR-1}
    For binary election instances $(\prof{B},\C)$ where $\C$ is $k$-restrictive for some $k$, $\cMES$ with fixed pricing $\lambda_{\text{fix}}$ satisfies $\nicefrac{1}{k}\text{-}{\agrEJR}\text{-}1$.
\end{theorem}

\begin{proof}{}
For a binary election instance $(\prof{B},\C)$ where $\C$ is $k$-restrictive, take an outcome $\vec{w}$ returned by $\cMES$ with fixed pricing $\lambda_{\text{fix}}$. Consider a $T$-agreeing voter group $N'$. Let us assume that for every $i\in N'$, it holds that $\sat{i}{\vec{w}} < \nicefrac{|N'|}{nk}\cdot|T|-1$ and then set $\ell = \nicefrac{|N'|}{nk}\cdot|T|-1$. So to conclude $\cMES$, each voter $i\in N'$ paid for at most $\ell-1 = \nicefrac{|N'|}{nk}\cdot|T|-2$ issues. Note that for a $k$-restrictive constraint $\C$, the maximum price $\cMES$ with fixed pricing $\lambda_{\text{fix}}$ sets for any issue-decision pair is $nk$ (as at most $k$ issues are fixed for a $\cMES$ purchase). Now, assume that the voters in $N'$ paid at most $\nicefrac{m}{(\ell + 1)}$ for any decision on an issue. We know that each voter has at least the following funds remaining at that moment:
\begin{equation*}
m - (\ell - 1)\frac{m}{\ell + 1}  = \frac{2m}{\ell + 1} = \frac{2m}{\nicefrac{|N'|}{kn}\cdot|T|} = \frac{2mnk}{|N'||T|} \geq \frac{2nk}{|N'|}.  
\end{equation*}

We now have that voter group $N'$ holds at least $2nk$ in funds at the rule's end. Thus, we know that at least $k$ issues have not been funded and for at least one of these $k$ issues, at least half of $N'$ agree on the decision for it (as the election is a binary instance) while having enough funds to pay for it. Hence, we have a contradiction to $\cMES$ terminating. 

Now, assume that some voter $i\in N'$ paid more than $\nicefrac{m}{(\ell + 1)}$ for fixing an issue's decision. Since we know that at the end of $\cMES$'s run, each voter in $N'$ paid for at most $\ell-1$ issues, then at the round~$r$ that voter $i$ paid more than $\nicefrac{m}{(\ell + 1)}$, the voters group $N'$ collectively held at least $2nk$ in funds. Since at least $k$ issues in were not funded, there exists some issue that could have been paid for in round~$r$, where voters each pay $\nicefrac{m}{(\ell + 1)}$. This contradicts the fact that voter $i$ paid more than $\nicefrac{m}{(\ell + 1)}$ in round $r$. 
Hence group of voters $N'$ cannot exist, concluding the proof.
\end{proof}

\noindent
 We now provide a definition of a constraint-aware version of the \emph{MeCorA} rule \citep{skowronGoreckiAAAI2022public}, in an effort to relax the assumption of working in a binary election instance. 
 In the unconstrained public-decision model, MeCorA is presented as an auction-style variant of MES that allows voter groups to change the decision of an issue all while increasing the price for any further change to this issue's decision. In our constrained version of this rule, voter groups are allowed to pay for changes to the decisions on \emph{sets of issues}, as long as these changes represent a feasible deviation.
 
 \begin{definition}[$\text{MeCorA}_{\C}$]
Take some constant $\epsilon > 0$. Start by setting $\lambda_{t} = 0$ as the current price of every issue $a_t\in \I$, endow each voter $i\in N$ with a personal budget of $m$ and take some arbitrary, feasible outcome $\vec{w}\in\C$ as the current outcome. A groups of voters can `update' the current outcome $\vec{w}$'s decisions on some issues $S\subseteq \I$ if the group:
\begin{itemize}
    \item[] $(i)$ can propose, for each issue $a_t\in S$, a new price $\lambda_{t}^{*} \geq \lambda_{t}+\epsilon$,
    \item[] $(ii)$ can afford the sum of new prices for issues in $S$, and
    \item[] $(iii)$ has an $(S,\vec{w},\C)$-deviation.
\end{itemize}
 The rule then works as follows. Given a current outcome $\vec{w}$, it computes, for every non-empty $S\subseteq \I$, the smallest possible value $\rho_{(t,S)}$ for each issue $a_t\in S$ such that for some $N'$, if voters in $N'$ each pay $\rho_{S} = \sum_{a_t\in S} \rho_{(t,S)}$ (or their remaining budget), then $N'$ is able to `update' the decisions on every $a_t\in S$ as per conditions $(i)-(iii)$. If there exists no such voter group for issues $S$ then it sets $\rho_{S} = \infty$. 

If $\rho_{S} = \infty$ for every $S\subseteq \I$, the process terminates and returns the current outcome $\vec{w}$. Otherwise, it selects the set $S$ with the lowest value $\rho_{S}$ (any ties are broken arbitrarily) and does the following:
\begin{enumerate}
    \item updates the current outcome $\vec{w}$'s decisions on issues in $S$ to the decisions agreed upon by the voters with the associated $(S,\vec{w},\C)$-deviation,
    \item  updates the current price of every issue $a_t\in S$ to $\lambda_{t}^{*}$,
    \item  returns all previously spent funds to all voters who paid for the now-changed decisions on issues in $S$,
    \item and finally, for each voter in $N'$, deduct $\sum_{a_t\in S} \rho_{(t,S)}$ from their personal budget (or the rest of their budget).
\end{enumerate}
\end{definition}

\begin{example}\label{exm:MeCorA}
      Suppose we have a election instance with three issues $\I = \{a_1, a_2, a_3\}$ with domains $D_{t} = \{d_{1},d_{2},d_{3}\}$ for all $t\in\{1,2,3,4\}$. Then take the constraint to be $\C  = \{(d_{1},d_{1},d_{1}),(d_{1},d_{1},d_{2}),(d_{1},d_{1},d_{3}),(d_{2},d_{2},d_{1})\}$. Then, suppose there are four voters $N = \{1,2,3\}$ with ballots $\vec{b}_{1} = (d_{1},d_{1},d_{1})$, $\vec{b}_{2} = (d_{1},d_{1},d_{2})$, $\vec{b}_{3} = (d_{2},d_{2},d_{2})$ and $\vec{b}_{4} = (d_{3},d_{3},d_{3})$.
      
      Assume that the rule begins with the current outcome $\vec{w} = (d_{2},d_{2},d_{1})$ where the price for each issue is $0$ and the individual voter budgets are $m=3$. As voters $\{1,2\}$ have an $(S,\vec{w},\C)$-deviation for the first two issues $S= \{a_{1},a_{2}\}$ and they each have their entire budgets available, they can pay to update the current outcome to $\vec{w}^{*} = (d_{1},d_{1},d_{1})$. This update raises the price of each of issues $\{a_1,a_2\}$ by $\epsilon$. In the next instance, there are feasible deviations on a single issue that are available for the groups $\{2,3\}$ (to outcome $(d_{1},d_{1},d_{2})$), $\{3\}$ (to outcome $(d_{2},d_{2},d_{1})$) and $\{4\}$ (to outcome $(d_{1},d_{1},d_{3})$). All of these deviations are affordable however, the rule favours the voter group $\{2,3\}$ as it is cheaper per voter. Thus, the current outcome is updated to $\vec{w}^{**} = (d_{1},d_{1},d_{2})$ and the price to update issue~$a_3$ is raised to $\epsilon$. Finally, suppose that at some future instance, the current outcome is again $\vec{w} = (d_{1},d_{1},d_{2})$ and suppose that the price of issue is at least $3$ (this is possible as each voter has a budget of $3$). Then, although voter~$4$ has an $(S,\vec{w},\C)$-deviation to outcome $(d_{1},d_{1},d_{3})$ and all of their budget (as they have zero satisfaction), they cannot afford to raise the price of issue~$a_3$ so cannot update the outcome.
\end{example}

\noindent
Our definition of $\text{MeCorA}_{\C}$ will be useful in Section~\ref{sec:priceable}, while now we need to refine its definition further to study its justified representation guarantees.
%
We first partition the voter population into groups where members of each group agree on some set of issues. Then, for each group, its members may only pay to change some decisions as a collective and only on those issues that they agree on. Contrarily to $\text{MeCorA}_{\C}$, voter groups cannot pay to change some decisions if this leads to the group's members gaining ``too much'' satisfaction from the altered outcome (i.e., a voter group exceeding their proportional share of their agreed-upon issues, up to some additive factor $q$ that parameterises the rule).

\begin{definition}[Greedy $\text{MeCorA}_{\C}\text{-}q$]
The set of the voters $N$ is partitioned into $p$ disjoints sets $N(T_{1}),\ldots,N(T_{p})$ such that:
\begin{itemize}
    \item[] $(i)$ for every $x\in \{1,\ldots,p\}$, a voter group $N(T_{x})\subseteq N$ is $T_{x}$-agreeing for some $T_{x}\subseteq \I$, and
    \item[] $(ii)$ for all $x\in\{1,\ldots,p-1\}$, it holds that $|N(T_{x})|\cdot |T_{x}| \geq |N(T_{x+1})|\cdot |T_{x+1}|$
\end{itemize}

As with $\text{MeCorA}_{\C}$, voter groups pay to change the decisions of some issues during the rule's execution. However, given the initial partition, during the run of Greedy $\text{MeCorA}_{\C}\text{-}q$, the voters in $N(T_{x})$ may only change decisions for the issues in $T_{x}$.
Moreover, if a voter group $N(T_{x})$ for some $x\in \{1,\ldots,p\}$ wishes to change some decisions at any moment during the process, this change does not lead to any voter in $N(T_{x})$ having satisfaction greater than $\nicefrac{|N(T_{x})|}{n}\cdot |T_{x}| - q$ with the updated outcome for some non-negative integer $q$. Besides these two differences, the rule works exactly as $\text{MeCorA}_{\C}$.
\end{definition}

\noindent
Consider the following example that demonstrates how Greedy $\text{MeCorA}_{\C}\text{-}q$ works. 

\begin{example}\label{exm:greedyMeCorA}
     Take an election instance with five issues $\I = \{a_1, a_2, a_3, a_4, a_5\}$ with domains $D_{t} = \{d_{1},d_{2},d_{3}\}$ for all $t\in\{1,2,3,4,5\}$. 
     Then, suppose there are five voters $N = \{1,2,3,4,5\}$ with ballots $\vec{b}_{1} = (d_{1},d_{1},d_{1},d_{1},d_{1})$, $\vec{b}_{2} = (d_{1},d_{1},d_{1},d_{1},d_{1})$, $\vec{b}_{3} = (d_{2},d_{2},d_{2},d_{1},d_{1})$, $\vec{b}_{4} = (d_{2},d_{2},d_{2},d_{2},d_{2})$ and $\vec{b}_{5} = (d_{3},d_{3},d_{3},d_{3},d_{3})$. The rule first partitions the voters into three groups $\{1,2\}$, $\{3,4\}$ and $\{4\}$ where the voter groups can only update the decisions of issues in $\{a_1, a_2, a_3, a_4,a_5\}$, $\{a_1, a_2,a_3,a_4\}$ and $\{a_1, a_2, a_3, a_4,a_5\}$, respectively. The only other difference to $\text{MeCorA}_{\C}$ is that for the voter groups $\{1,2\}$, $\{3,4\}$ and $\{4\}$, the satisfaction garnered from the new outcome in any deviating update cannot exceed a certain amount. Suppose that $q=0$ for Greedy $\text{MeCorA}_{\C}\text{-}q$. For voter groups $\{1,2\}$, $\{3,4\}$ and $\{4\}$, this satisfaction limit is set to $\nicefrac{2}{5}\cdot5=2$, $\nicefrac{2}{5}\cdot3 = \nicefrac{6}{5}$ and $\nicefrac{1}{5}\cdot5 = 1$, respectively.
\end{example}

\noindent
We can now show our main result for Greedy $\text{MeCorA}_{\C}\text{-}q$ on a $k$-restrictive constraint. For this result, we require that voter ballots are feasible outcomes in $\C$, but we do not need to assume that we are in a binary election instance.

\begin{theorem}\label{thm:greedyMeCorA_approx_satisfies_agrEJR}
For election instances $(\prof{B},\C)$ where voters' ballots are consistent with the constraint $\C$ and $\C$ is $k$-restrictive for some $k\geq 2$, Greedy $\text{MeCorA}_{\C}\text{-}(k-1)$ satisfies ${\agrEJR}\text{-}(k-1)$.
\end{theorem}

\begin{proof}{} 
Take an outcome $\vec{w}$ returned by Greedy  $\text{MeCorA}_{\C}\text{-}(k-1)$. Assume that $\vec{w}$ does not provide ${\agrEJR}\text{-}(k-1)$. Thus, there is a $T$-agreeing group $N'$ such that $\sat{i}{\vec{w}} < \nicefrac{|N'|}{n}\cdot |T| - k+1 = \ell$ holds for every $i\in N'$. 
Now, consider the partition of voters $N(T_{1}),\ldots,N(T_{p})$ constructed by Greedy $\text{MeCorA}_{\C}\text{-}(k-1)$ to begin its run. 
Assume first that there is some $x\in \{1,\ldots,p\}$ such that $N' = N(T_{x})$, i.e., voters $N'$ appear in their entirety in said partition. 
We then have $T= T_x$. Moreover, voters in $N'$ each contribute to at most $\ell$ decisions at any moment of the run of Greedy $\text{MeCorA}_{\C}\text{-}(k-1)$, as this is the limit the rule imposes on their total satisfaction. We now consider two cases. 
Assume that the voters in $N'$ contributed at most $\nicefrac{m}{(\ell + k - 1)}$ to change some decisions during the rule's execution. It follows that each voter has at least the following funds remaining: $m - (\ell - 1)\cdot\nicefrac{m}{(\ell + k - 1)} \geq \nicefrac{nmk}{|N'||T|}$.

In this case, the voters in $N'$ would have at least $\nicefrac{nmk}{|T|}$ in collective funds, so it follows that each distinct $(S,\vec{w},\C)$-deviation available to $N'$ must cost at least $\nicefrac{nmk}{|T|}$. As $N'$ is $T$-agreeing, it must be that $N'$ has at least a $\nicefrac{(|T| - \ell + 1)}{k}$ many $(S,\vec{w},\C)$-deviations due to $\C$ being $k$-restrictive and as the voters' ballots are consistent with $\C$.

Now, consider the case where some voter in $N'$ contributed more than $\nicefrac{m}{(\ell + k - 1)}$ to change some decisions. The first time that this occurred, the change of decisions did not lead to any voter in $N'$ obtaining a satisfaction greater than $ \ell = \nicefrac{|N'|}{n}\cdot |T| - k + 1$ (otherwise the rule would not allow these voters to pay for the changes). Thus, each voter in $N'$ must have contributed to at most $\ell -1$ issues before this moment. From the reasoning above, it must hold that in this moment, each voter held at least $\nicefrac{nmk}{|N'||T|}$ in funds with there being at least $\nicefrac{(|T| - \ell + 1)}{k}$ feasible deviations available to $N'$ and each such deviation costing at least $\nicefrac{nmk}{|T|}$. So in both cases, for the $(S,\vec{w},\C)$-deviations that are present in $T$ that voters in $N'$ wish to make, outcome $\vec{w}$'s decisions must have been paid for by voters within the remaining voter population $N\setminus N'$. And so, these decisions must have cost the voters in $N\setminus N'$ at least:
\[
\frac{nmk}{|T|}\cdot\bigg(\frac{|T| - \ell + 1}{k}\bigg) = \frac{nm}{|T|}\cdot\bigg(|T| - \frac{|N'|}{n}\cdot |T|  + k\bigg)
\]
\[
> \frac{nm}{|T|}\cdot\bigg(\frac{n|T| - |N'||T|}{n}\bigg) = m(n - |N'|).
\]
However, voters $N\setminus N'$ have at most $m(n - |N'|)$ in budget. Thus, the rule cannot have terminated with the voter group $N'$ existing.

Now, assume that the group $N'$ did not appear in their entirety within the partition $N(T_{1}),\ldots,N(T_{p})$ made by Greedy $\text{MeCorA}_{\C}\text{-}(k-1)$. This means that some voter $i\in N'$ is part of another voter group $N(T_{x})$ that is $T_{x}$-agreeing such that $\nicefrac{|N(T_{x})|}{n}\cdot |T_{x}| \geq \nicefrac{N'}{n}\cdot |T|$. Now, recall that for each voter group $N(T)$ in the partition, the voters in $N(T)$ have the same satisfaction to end the rule's execution (as they only pay to flip decisions as a collective). Thus, from the arguments above, it holds for this voter $i\in N'\cap N(T_{x})$ that $\sat{i}{\boldsymbol{w}} \geq \nicefrac{|N(T_{x})|}{n} \cdot |T_{x}| - k +1\geq \nicefrac{|N'|}{n}\cdot |T| - k +1$, which contradicts the assumption that every voter in $N'$ has satisfaction less than $\nicefrac{|N'|}{n}\cdot |T| - k +1$. 
\end{proof}

\noindent
We now explore another direction towards producing proportional outcomes on $k$-restrictive constraints. 
We define a constraint-aware version of the \emph{Local Search Proportional Approval Voting (LS-PAV)} rule, which is a polynomial-time computable rule that is known to satisfy EJR \citep{azizEtAlSCW2017justified}. 
In the committee elections setting, the rule begins with an arbitrary committee of some fixed size $k$, and improves it in iterations by searching for any swaps between committee members and non-selected candidates that brings about an increase of the \emph{PAV score} by at least $\nicefrac{n}{k^{2}}$. In our model, the PAV score of some feasible outcome  $\vec{w}\in \C$ is defined to be $\text{PAV}(\vec{w}) = \sum_{i\in N} \sum_{t = 1}^{\sat{i}{\vec{w}}} \nicefrac{1}{t}$. 

\begin{definition}[Local Search $\text{PAV}_{\C}$, LS-$\text{PAV}_{\C}$]
Let $\vec{w}\in \C$ be an arbitrary outcome. 
If there exists an ($S$,$\vec{w}$)-deviation for some voter group to some outcome $\vec{w}'\in\C$ such that $\text{\sc PAV}(\vec{w}') - \text{\sc PAV}(\vec{w})\geq \nicefrac{n}{m^{2}}$, i.e., the new outcome $\vec{w}'$ yields a PAV score that is at least $\nicefrac{n}{m^{2}}$ higher than that of $\vec{w}$, then the rule sets $\vec{w}'$ as the current winning outcome. The rule proceeds in iterations until there exists no deviation that improves the PAV score of the current winning outcome by at least $\nicefrac{n}{m^{2}}$.
\end{definition}

\begin{example}
    Suppose we have a election instance with three issues $\I = \{a_1, a_2, a_3\}$ with domains $D_{t} = \{d_{1},d_{2},d_{3}\}$ for all $t\in\{1,2,3,4\}$. Then take the constraint to be $\C  = \{(d_{1},d_{1},d_{2}),(d_{2},d_{2},d_{2}),(d_{3},d_{1},d_{2})\}$. Then, suppose there are four voters $N = \{1,2,3\}$ with ballots $\vec{b}_{1} = (d_{1},d_{1},d_{1})$, $\vec{b}_{2} = (d_{1},d_{1},d_{1})$, $\vec{b}_{3} = (d_{2},d_{2},d_{2})$ and $\vec{b}_{4} = (d_{3},d_{3},d_{3})$. So LS-$\text{PAV}_{\C}$ updates the current outcome if the PAV score increases by at least $\nicefrac{n}{m^2} = \nicefrac{4}{9}$. Suppose that the rule begins with $\vec{w} = (d_{2},d_{2},d_{2})$ being the current outcome. Observe that updating the current outcome to $\vec{w}^{*} = (d_{1},d_{1},d_{2})$ (which is a valid an ($S$,$\vec{w}$)-deviation for voter group $\{1,2\}$) would add $3$ to the PAV score (voters $1$ and $2$ being satisfied by their first two decisions) while taking away $\nicefrac{5}{6}$ (from voter~$3$'s lost agreement on issues $\{1,2\}$) which is more than $\nicefrac{4}{9}$, so this update occurs. Consider in the next instance that voter~$4$ wishes to deviate to outcome $\vec{w}^{**} = (d_{3},d_{1},d_{2})$. This would result in no change to the PAV score (as voter~$4$'s added contribution of $1$ is cancelled out by voter $1$ and $2$'s lost contribution of $\nicefrac{1}{2}$ each) and thus, the rule does not make this change. In fact, $\vec{w}^{*} = (d_{1},d_{1},d_{2})$ yields the largest PAV score is returned as the final outcome by LS-$\text{PAV}_{\C}$.
\end{example}

\noindent
As there is a maximum obtainable PAV score, LS-$\text{PAV}_{\C}$ is guaranteed to terminate. The question is how long this rule takes to return an outcome when we have to take $k$-restrictive constraints into account. 

\begin{proposition}
    For elections instances where $\C$ is $k$-restrictive (where $k$ is a fixed constant), LS-$\text{PAV}_{\C}$ terminates in polynomial time. 
\end{proposition}

\begin{proof}{}
We show that given an outcome $\vec{w}$, finding all possible deviations can be done in polynomial time for a $k$-restrictive constraint $\C$. This can be done by exploiting the presence of the implication set $\textit{Imp}_{\C}$. Note that the size of the implication set $\textit{Imp}_{\C}$ is polynomial in the number of issues. So we can construct the outcome implication graph of $\textit{Imp}_{\C}$ and the outcome $\vec{w}$ in polynomial time. Then for each issue $a_t\in \I$, we can find the set $G_{\vec{w}}(a_t,d)$ for some $d\neq w_{t}\in D_{t}$ in polynomial time and the issue-decision pairs represent the required additional decisions to be fixed in order to make a deviation from outcome $\vec{w}$ by changing the $\vec{w}$'s decision on issue $a_t$ to $d$. Doing this for each issue $a_t$ allows us to find a deviation that can improve the PAV score, if such a deviation exists. With similar reasoning used in other settings \citep{azizEtAlSCW2017justified,chandakARXIV2024sequential}, we end by noting that since there is a maximum possible PAV score for an outcome, and each improving deviation increases the PAV score by at least $\nicefrac{n}{m^2}$, the number of improving deviations that LS-$\text{PAV}_{\C}$ makes is polynomial in the number of issues $m$. 
\end{proof}

\noindent
Off the back of this positive computational result, we present the degree to which LS-$\text{PAV}_{\C}$ provides proportional outcomes with regards to the $\alpha\text{-}{\agrEJR}\text{-}\beta$ axiom. 

\begin{theorem}\label{thm:lsPAV_approx_satisfies_agrEJR}
   For election instances $(\prof{B},\C)$ where the voters' ballots are consistent with the constraint $\C$ and $\C$ is $k$-restrictive for some $k\geq 2$, LS-$\text{PAV}_{\C}$ satisfies $\nicefrac{2}{(k+1)}\text{-}{\agrEJR}\text{-}(k-1)$.
\end{theorem}

\begin{proof}{}
For an election instance $(\prof{B},\C)$ where $\C$ is $k$-restrictive for $k\geq 2$, take an outcome $\vec{w}$ returned by LS-$\text{PAV}_{\C}$ and consider a group of voters $N'$ that agree on some set of issues $T$. Let us assume that for every voter $i\in N'$, it holds that $\sat{i}{\vec{w}} < \nicefrac{2}{k+1}\cdot\nicefrac{|N'|}{n}\cdot|T|-k + 1$ and then set $\ell = \nicefrac{2}{k+1}\cdot\nicefrac{|N'|}{n}\cdot|T|-k + 1$. We use $r_{i}$ to denote the number of outcome $\vec{w}$'s decisions that a voter $i\in N$ agrees with.

For each voter $i\in N\setminus N'$, we calculate the maximal reduction in PAV score that may occur from a possible deviations by LS-$\text{PAV}_{\C}$ when $\C$ is $k$-restrictive. This happens when for each of at most $\nicefrac{r_{i}}{k}$ deviations, we decrease their satisfaction by $k$ and remove $\sum_{t=0}^{k-1}\nicefrac{1}{(r_{i} - t)}$ in PAV score. So for these voters in  $N\setminus N'$, we deduct at most the following:
\[
 \sum_{N\setminus N'} \frac{r_{i}}{k} \cdot \bigg(\sum\limits_{t=0}^{k-1}\frac{1}{r_{i} - t}\bigg) \leq \sum_{N\setminus N'} \frac{r_{i}}{k} \cdot\bigg(\frac{\sum_{t=1}^{k}t}{r_{i}}\bigg) = \frac{k+1}{2} \cdot(n - |N'|).
\]

Now, so there are $|T| - (\ell - 1) = |T| - \ell +1$ issues that all voters in $N'$ agree on but they disagree with outcome $\vec{w}$'s decisions on these issues. Since we assume the constraint is $k$-restrictive, then for each of these $|T| - \ell +1$ issues, they fix at most $k-1$ other issues and thus, there are at least $\nicefrac{(|T| - \ell +1)}{k}$ feasible deviations that can be made by LS-$\text{PAV}_{\C}$ amongst these issues. For the voters in $N'$, we now consider the minimal increase in PAV score that may occur from these possible deviations by LS-$\text{PAV}_{\C}$. For each such deviation, we increase their satisfaction by at least $k$ and thus, for a voter $i\in N'$, we increase the PAV score by $\sum_{t=1}^{k}\nicefrac{1}{(r_{i} + t)}$. Since for each voter $i\in N'$ we have $r_{i} \leq \ell - 1$, and as there are at least $\nicefrac{(|T| - \ell +1)}{k}$ feasible deviations in $T$, it follows that we add at least the following to the PAV score: 

\[
\frac{|T| - \ell +1}{k} \cdot\bigg(\sum_{i\in N'} \sum\limits_{t=1}^{k}\frac{1}{r_{i} + t}\bigg)\geq \frac{|T| - \ell +1}{k} \cdot\bigg(\sum_{i\in N'} \sum\limits_{t=1}^{k}\frac{1}{\ell + t - 1}\bigg)
\]

Taking into account that $k\geq 2$ and $\ell = \nicefrac{2|N'||T|}{(n(k+1))}-k+1$, then with further simplification, we find that at least the following is added to the PAV score:
\[
> \frac{n(k+1)}{2} - |N'| +  \frac{n(k+1)}{|T|} \geq \frac{k+1}{2} \cdot(n - |N'|) +  \frac{n(k+1)}{|T|}.
\]

So the total addition to the PAV score due to satisfying voters in $N'$ is strictly greater than the PAV score removed for the added dissatisfaction of voters in $N\setminus N'$ (which is at most $\nicefrac{(k+1)(n - |N'|)}{2}$). And specifically, this change in score is at least $\nicefrac{n(k+1)}{|T|} > \nicefrac{n}{|T|}$ and thus, at least one of the $\nicefrac{(|T| - \ell +1)}{k}$ many deviations must increase the PAV score by more than:
\[
\frac{k}{|T| - \ell +1}\cdot \frac{n}{|T|} \geq \frac{1}{|T|}\cdot \frac{n}{|T|} \geq \frac{n}{|T|^{2}} \geq \frac{n}{m^{2}}.
\]

Thus, LS-$\text{PAV}_{\C}$ would not terminate but would instead make this deviation in order to increase the total PAV score. Thus, contradicting that such a group $N'$ cannot exist. 
\end{proof}

\noindent
With this result, we have a rule that when focused on $k$-restrictive constraints, is both polynomial-time computable and provides substantial proportional representation guarantees (assuming voter ballots are constraint consistent).

\section{Proportionality via Priceability}\label{sec:priceable}

In this section we propose an alternative to the interpretation of proportional representation as justified representation through the notion of priceability \citep{petersSkowronEC2020welfarism,LacknerS23mwv}. Recent work has shown the promise of this market-based approach for a general social choice model \citep{MasarikP024generalProp} and the sequential choice model \citep{chandakARXIV2024sequential}. 
We propose the following version for constrained public decisions albeit looking at a weaker priceability axiom than that of \citet{MasarikP024generalProp} as they use an adaptation of \emph{stable priceability} \citep{petersPSS2021market} that is not always satisfiable \citep{MasarikP024generalProp} (see Section~$6.2$ in their paper).

\begin{definition}[Priceability]
Each voter has a personal budget of $m$ and they have to collectively fund the decisions on some issues, with each decision coming with some price. A \emph{price system} $\textbf{ps} = (\{p_{i}\}_{i\in N},\{\pi_{(a_t,d)}\}_{(a_t,d)\in H})$ where $H = \bigcup_{a_{t}\in \I}\{(a_{t},d)\mid d\in D_{t}\}$ is a pair consisting of $(i)$ a collection of \emph{payment functions} $p_{i}: \I\times\{0,1\} \rightarrow [0, b]$, one for each voter $i\in N$, and $(ii)$ a collection of \emph{prices} $\pi_{(a_t,d)} \in \mathbb{R}_{\geq 0}$, one for each decision pair $(a_t,d)$ for $a_t\in\I$ and $d\in D_{t}$. We consider priceability with respect to outcomes $\vec{w} \in\C$ where decisions are made on all issues. 
We say that an outcome $\vec{w} = (w_{1},\ldots,w_{m})$ is \emph{priceable} if there exists a price system $\textbf{ps}$ such that:
 
 \begin{enumerate}
     \item[] $(\rm P1):$ For all $a_t\in\I$ and $d\in D_{t}$, it holds that if $d \neq \vec{b}_{i,t}$ we have $p_{i}(a_t,d) = 0$, for every $i\in N$.
     
     \item[] $(\rm P2):$ $\sum_{(a_t,d)\in H} p_{i}(a_t,d) \leq m$ for every $i\in N$ where it holds that $H = \bigcup_{a_{t}\in \I}\{(a_{t},d)\mid d\in D_{t}\}$.
     
    \item[] $(\rm P3):$ $\sum_{i\in V} p_{i}(a_t,d) = \pi_{(a_t,w_{t})}$ for every $a_t\in\I$.
    
    \item[] $(\rm P4):$ $\sum_{i\in V} p_{i}(a_t,d) = 0$ for every $a_t\in\I$ and every $d\neq w_{t}\in D_{t}$.

    \item[] $(\rm P5):$ There exists no group of voters $N'$ with an $(S,\vec{w},\C)$-deviation for some $S\subseteq \I$, such that for each $a_t\in S$:
    \[
    \sum\limits_{i\in N'} \bigg(m-\sum\limits_{(a_{t}',d')\in H} p_{i}(a_{t}',d')\bigg) > \pi_{(a_t,w_{t})}
    \] 
    where $H = \bigcup_{a_{t}\in \I}\{(a_{t},d)\mid d\in D_{t}\}$.
\end{enumerate}
    
\end{definition}

\noindent
Condition (P1) states that each voter only pays for decisions that she agrees with; (P2) states that a voter does not spend more than her budget $m$; (P3) states that for every decision in the outcome, the sum of payments for this decision is equal to its price; (P4) states that no payments are made for any decision not in the outcome; and, finally, (P5) states that for every set of issues $S$, there is no group of voters $N'$ agreeing on all decisions for issues in $S$, that collectively hold more in unspent budget to `update' outcome $\vec{w}$'s decision on every issue $a_t\in S$ to a decision that they all agree with (where `updating' these issues leads to a feasible outcome). We illustrate priceability in our setting with the following example of a binary election instance.

\begin{example}
    Take four issues $\I = \{a_1, a_2, a_3, a_4\}$ in a binary election instance and a constraint $\C = \{(1,1,1,1),(1,1,0,0)\}$. Suppose there are two voters with ballots $\vec{b}_{1} = (1,1,1,1)$ and  $\vec{b}_{2} = (0,0,0,0)$. Note that outcome $\vec{w} = (1,1,1,1)$ is not priceable as any price system where voter $1$ does not exceed her budget would have voter $2$ having enough in leftover budget to cause a violation of condition $(\rm P5)$ (with her entire budget being leftover, she can afford more than the price of the $(S,\vec{w},\C)$-deviation to outcome $\vec{w}$). On the other hand, $\vec{w}' = (1,1,0,0)$ is priceable where we set the price of this outcome's decisions to $1$.
\end{example}

\noindent
The following result gives some general representation guarantees whenever we have priceable outcomes.

\begin{proposition}\label{prop:c-priceable-PJR-like-rep}
    Consider a priceable outcome $\vec{w}$ with price system $\textbf{ps} = (\{p_{i}\}_{i\in N},\{\pi_{(a_t,d)}\}_{(a_t,d)\in H})$ where $H = \bigcup_{a_{t}\in \I}\{(a_{t},d)\mid d\in D_{t}\}$. Then, for every $T$-cohesive group of voters $N'\subseteq N$ for some $T\subseteq \I$ with an $(S,\vec{w},\C)$-deviation for some non-empty $S\subseteq T$, it holds that:
    \[
    \sum\limits_{i\in N'} \sat{i}{\vec{w}} \geq \frac{n}{q}\cdot|T| - |S|
    \]
    where $q = \max\{\pi_{(a_t,w_{t})}\}_{a_t\in S}$.
\end{proposition}

\begin{proof}{}
Take a priceable outcome $\vec{w}$ and consider a $T$-cohesive group of voters $N'$. Suppose that $\sum_{i\in N'} \sat{i}{\vec{w}}< \nicefrac{n}{q}\cdot|T| - |S|$ where $q = \max\{\pi_{(a_t,w_{t})}\}_{a_t\in S}$. As a group, the voters $N'$ have a budget of $m|N'|$. Now, the voters in $N'$ collectively contributed to at most $\nicefrac{n}{q}\cdot|T| - |S|-1$ decisions in outcome $\vec{w}$, and for each decision, the price was at most $q$ (as $q$ is the price system's maximal price). So, we have that voter group $N'$ has at least the following in leftover budget:
\[
m|N'| - q\cdot\bigg(\frac{n}{q}\cdot|T| - |S|-1\bigg)\geq m\cdot \frac{n|T|}{m} - n|T| + q|S|+q = q\cdot(|S|+1).
\]
Note we made use of the fact that $N'$ is $T$-cohesive.
Thus, we know that $N'$ has strictly more than $q|S|$ in funds and for each issue in $a_t\in S$, holds more than in funds than $q\geq \pi_{(a_t,w_{t})}$. This presents a violation of condition ${\rm P5}$ of priceability. Hence, voter group $N'$ cannot exist.
\end{proof}

\noindent
However, we now must ascertain whether priceable outcomes always exist, regardless of the nature of the constraint. We see that this is possible thanks to the rule we have already defined, namely $\text{MeCorA}_{\C}$.

The next result shows that $\text{MeCorA}_{\C}$ captures the notion of priceability.

\begin{proposition}\label{prop:MeCorA-c_is_priceable}
    $\text{MeCorA}_{\C}$ always returns priceable outcomes.
\end{proposition}

\begin{proof}{}
Let $\vec{w} = (w_{1},\ldots,w_{m})$ be the outcome returned by $\text{MeCorA}_{\C}$. We define the following price system $\textbf{ps}$: For each issue $a_t\in\I$, fix the prices $\pi_{(a_t,w_{t})} = \pi_{(a_t,d)} = \lambda_{t}$ for all $d\neq w_{t}\in D_{t}$ where $\lambda_{t}$ is issue $a_t$'s last $\text{MeCorA}_{\C}$ price (before being set to $\infty$) prior to the rule's termination. Fix the payment functions $p_{i}$ for each voter to the money they spent to end the execution of $\text{MeCorA}_{\C}$. Observe that the priceability conditions $(\rm P1)$-$(\rm P4)$ clearly hold: since we have that, to end $\text{MeCorA}_{\C}$'s run, voters do not pay for decisions that $(i)$ they do not agree with (condition $(\rm P1)$) and $(ii)$ are not made by outcome $\vec{w}$ (condition $(\rm P4)$); $\text{MeCorA}_{\C}$ limits each voter a budget of $m$ (condition $(\rm P2)$) $(\rm P2)$; and the sum of payments for decisions made by outcome $\vec{w}$ will equal exactly $\pi_{(a_t,w_{t})} = \lambda_{t}$ (condition $(\rm P3)$). Now, for condition $(\rm P5)$, note that if such a group of voters $N'$ existed for some set of issues $S$, then $\text{MeCorA}_{\C}$ would not have terminated as this group of voters could have changed the decisions of these issues in $S$ while increasing each issues' prices.
\end{proof}

\noindent
This is a positive result that, combined with that of Proposition~\ref{prop:c-priceable-PJR-like-rep}, gives a rule that always returns us priceable outcomes for any election instance.

\section{Conclusion}\label{sec:conclusion}

We considered two different interpretations of justified representation from committee elections and adapted them to a public-decision model with constraints. 
In analysing the feasibility of the axioms, we devised restricted classes of constraints such as those not fixing decisions and those that can be represented as simple implications. 
While we could show mostly negative results for the satisfaction of cohesiveness-EJR under constraints, we were able to adapt successfully three known rules (MES, Local Search PAV and MeCorA) to yield positive proportional guarantees that meet, in an approximate sense, the requirements of agreement-EJR. Additionally, we defined a suitable notion of priceability and showed that our adaptation of MeCorA always returns priceable outcomes.

Our work opens up a variety of paths for future research. First, assessing a class of constraints that are more expressive than simple implications seems a natural starting point in extending our work. Then, on a more technical level, it would be interesting to check if the representation guarantees that are offered by $\cMES$, LS-$\text{PAV}_{\C}$ and Greedy $\text{MeCorA}_{\C}\text{-}(k-1)$ still hold for a wider range of election instances. Regarding our adaptation of priceability, the question is open as to whether there are more constrained public-decision rules that always produce complete priceable outcomes. Given that we opted to represent the constraints as an enumeration of all feasible outcomes, it is natural to ask whether our computational results 
still hold under compact constraint representations, e.g., $\C$ is represented as a Boolean formula of propositional logic. We also note some lingering computational questions such as the computational complexity of $(i)$ computing outcomes for rules such as $\cMES$ and Greedy $\text{MeCorA}_{\C}\text{-}(k-1)$ for general constraints, and $(ii)$ of checking whether a given feasible outcome is priceable. Finally, the list of proportionality notions to be tested on the constraints test-bed is not exhausted \citep{LacknerS23mwv}, with the proportionality degree \citep{skowron21propDegree} and $\text{EJR}{+}$ \citep{BrillP20023robust} being notable notions still to be considered. 

\clearpage
\bibliographystyle{abbrvnat}
\bibliography{theorydec}  

\end{document}